\documentclass[11pt,reqno]{amsart}

\usepackage{amscd,amssymb,amsmath,amsthm}
\usepackage{graphicx}
\usepackage{color}
\usepackage{cite}
\newtheorem{thm}[subsection]{Theorem}
\newtheorem{lemma}[subsection]{Lemma}
\newtheorem{pro}[subsection]{Proposition}
\newtheorem{cor}[subsection]{Corollary}

\theoremstyle{definition}
\newtheorem{rk}[subsection]{Remark}

\numberwithin{equation}{section}

\def\s{\sigma}

\def\s{\sigma}

\def\s{\sigma}

\def\R{\mathbb{R}}

\begin{document}
\title[Extreme Gibbs measures for the Potts model]{
Fuzzy transformations and extremality of Gibbs measures
for the Potts model on a Cayley tree
}

\author{C. Kuelske, U. A. Rozikov}

\address{C.\ Kuelske\\ Fakult\"at f\"ur Mathematik,
Ruhr-University of Bochum, Postfach 102148,\,
44721, Bochum,
Germany}
\email {Christof.Kuelske@ruhr-uni-bochum.de}

\address{U.\ A.\ Rozikov\\ Institute of mathematics,
29, Do'rmon Yo'li str., 100125, Tashkent, Uzbekistan.}
\email {rozikovu@yandex.ru}

\begin{abstract} We continue our study of the full set of translation-invariant splitting Gibbs measures
(TISGMs, translation-invariant tree-indexed Markov chains) for the $q$-state Potts model on a Cayley tree.
In our previous work \cite{KRK} we gave a full description of the TISGMs, and showed
 in particular that at sufficiently low temperatures their number is  $2^{q}-1$.
 In this paper we find some regions for the temperature parameter ensuring that a given TISGM is (non-)extreme
 in the set of all Gibbs measures.
 In particular we show the existence of a temperature interval for which there are at least
$2^{q-1} + q$ extremal TISGMs.
 For the Cayley tree of order two we give explicit formulae and some numerical values.
\end{abstract}
\maketitle



\section{Introduction}

For the $q$-state Potts model on a Cayley tree of order $k\geq 2$ it has been known for a long time that
at sufficiently low temperatures there are at least $q+1$ translation-invariant Gibbs measures
which are also tree-indexed Markov chains \cite{Ga8}, \cite{Ga9}.
Such translation-invariant
tree-indexed measures are equivalently called  translation-invariant splitting Gibbs measures (TISGMs).
One of the $q+1$ well-known measures mentioned above is obtained as
the infinite-volume limit of finite dimensional Gibbs measures with free
boundary conditions and each of the remaining $q$ measures is obtained as
the limit with homogeneous (constant) spin-configurations
as boundary conditions (see the next section for more details).

It is known (see, e.g., \cite{Ge}) that for all
$\beta >0$, the Gibbs measures  form a non-empty
convex compact set in the space of probability measures.
Extreme measures, i.e. extreme points of this set are
associated with pure phases.
Furthermore, any Gibbs measure is an integral of
extreme ones (the extreme decomposition).

The $q$ measures
with homogeneous boundary conditions are always extremal in the set of all Gibbs-measures \cite{Ga8}, \cite{Ga9},
the free boundary condition measure is  an extremal Gibbs measure only
in an intermediate temperature interval below the transition temperature,
and loses its extremality for even lower temperatures \cite[Theorem 5.6.]{Ro}.

The non-trivial problem to determine the precise transition
value is also equivalent to a reconstruction problem in information theoretic
language.
Indeed, the reconstruction problem is well-defined for
any tree-indexed Markov chain defined in terms of a given
transition matrix. The states we will consider are specific examples of such Markov chains.
Single out a root vertex and
denote by $\Pi^a_n$ the distribution of spins at distance
$n$ from the root which is obtained by conditioning the spin at the root to be $a$.
One says that the measure is reconstructable
if there is a pair of symbols $a,b$ such that the total variational distance
between $\Pi^a_n$ and $\Pi^b_n$ stays bounded against zero as the distance
$n$ tends to infinity. In other words, some information about the symbol at the root survives
a broadcasting even over arbitrary far distances,
and this property is equivalent to non-triviality of the tail-sigma
algebra, see \cite{Mos}. It is well-known that  non-triviality of the tail-sigma
algebra is equivalent to non-extremality.

While the Kesten-Stigum bound gives a sufficient (but not necessary) bound in terms of the second largest eigenvalue of
 the transition matrix of the chain (and hence the temperature parameter) which ensures non-extremality \cite{Ke},\cite{Sly11}, the opposite problem,
namely to ensure extremality in terms of bounds on parameters, is more difficult, and the known methods
usually do not lead to sharp results for tree indexed Markov chains.
This is true even in simple cases like the general asymmetric binary channel
(or Ising model in a magnetic field).  For the Potts model in zero field it was proved in \cite{Sly11}
that the Kesten-Stigum bound for the free boundary condition measure
is not sharp for $q\geq 5$. For some numerical investigations
see \cite{MM}.

For other results related to the Potts model and TISGMs on Cayley trees see \cite{Br1},
\cite{Ha96}, \cite{HaKu04}, \cite{KRK}, \cite{Ro} and the references therein.

Now, in \cite{KRK}  all TISGMs (tree-indexed Markov chains) for the Potts model are found on the Cayley tree
of order $k\geq 2$, and it is shown that at sufficiently low
temperatures their number is  $2^{q}-1$. In the case $k=2$ the explicit formulae for the
critical temperatures and all TISGMs are given.

The analysis was based on the classification
 of translation-invariant boundary laws  which are in one-to-one correspondence with the TISGMs.
The boundary laws are length-$q$ vectors with positive entries which satisfy a
non-linear fixed-point equation (tree recursion), and a given boundary law defines
the transition matrix of the corresponding Markov chain (see the next section for detailed definitions).
 It is an interesting but different issue
to determine natural sets of boundary conditions of
{\em spin configurations} for which the new states appear as finite volume limits. When chosen
deterministically, such boundary conditions could be periodic but non-translation invariant,
or they could be chosen from a natural measure to be constructed.

While the fact that these measures can never be nontrivial convex combinations of each other
(i.e. extremal in the set of all TISGMs)
is almost automatic (see  \cite[Theorem 2]{KRK}) it is not clear whether and when they are extremal
in the set of all Gibbs measures, including the non-translation invariant Gibbs measures.

In particular, from \cite{KRK} it was not clear yet, whether the complete set of TISGMs
contained any new extremal Gibbs measures beyond the known $q+1$ measures, or whether
the new TIGSMs would all be non-extremal Gibbs measures.
It is the purpose of this paper to give answers
and find regions of parameters where a TISGM is (non-)extreme. To ask immediately for precise transition values,
would be too much to ask for, in view of the known difficulty of the aforementioned reconstruction problem
already for the free boundary condition measure.

The paper is organized as follows. Section 2 contains preliminaries (necessary definitions
and facts).
All of our analysis relies heavily on the following useful fact: While the full permutation symmetry of
the free boundary condition measure is lost in general,
all the $q\times q$ transition matrices which arise in the description of the TISGMs possess a $2\times 2$ block-structure.
All possible sizes of blocks can appear, and labels within the blocks are equivalent.  This
corresponds to a decomposition of the $q$ spin-labels into two classes $1'$ and $2'$, one of $m$ elements,
the other one of $q-m$ elements, with $m\leq q/2$\footnote{See \cite[Corollary 2]{KRK} for the reason why we can assume that $m\leq q/2$. Indeed, the measures corresponding to $m>q/2$ will coincide, by the symmetry of the model,
with other measures given for $m\leq q/2$.}. Such a structure
invites the study of the associated fuzzy map (see \cite{Ha03}, \cite{HaKu04})
which identifies the spin-values within the classes, and which maps the Markov chain to a coarse-grained
or fuzzy  Markov chain with spin labels $1',2'$. It is interesting to note that this Markov chain can be interpreted as
a splitting Gibbs measure for an effective Ising Hamiltonian with an external field, and that this Hamiltonian is independent of the choice of the Gibbs measure
(within the class indexed by $m$).
Such a result, namely the independence of the Hamiltonian for a transformed (renormalized) Gibbs measure
of the phase, if this renormalized Hamiltonian exists at all,  is known to be true for lattice models \cite{EnFeSo93, Ku04}
with proofs which are based on the Gibbs variational
principle which holds for lattice systems but does not exist on trees.\footnote{
For lattice measures
there is the well-defined notion of relative entropy density
$h(\mu|\nu)$ between two infinite-volume measures.  $h(\mu|\nu)$ is obtained as a finite-volume
limit. It is well-defined if  $\nu$ is a translation-invariant Gibbs measure,
and $\mu$ is any translation-invariant measure.
The part  of the variational principle in the infinite volume which is relevant here
states the following:
If $\nu$ is a Gibbs measure, $\mu$ is a
translation-invariant measure, and $h(\mu|\nu)=0$ then
$\mu$ is also Gibbs measure for the same specification as $\nu$ (see \cite{EnFeSo93}).
On trees boundary-term contributions to the energy for finite subtrees
are of no smaller order than the volume terms. This causes problems for the definition
of the relative entropy density, and the variational principle does not exist.}

Furthermore, \cite{EE} provides an explicit counterexample for such a theorem to hold in full generality, by showing  that the Gibbs property of a time-evolved Ising tree measure does depend  on the phase of the Ising model, in certain parameter regimes.
It would be very interesting nevertheless to understand better to what extent the phenomenon of independence of the renormalized Hamiltonian of the phase for measures on trees we find here might generalize beyond the case we find in our example.

Section 3 is devoted to this fuzzy transformation of the Potts model to an Ising model with an external field.
There we give a decomposition of the Potts measures in terms of their
image measures under fuzzy transformation, and the conditional measure.
We also study relations between eigenvalues
of the relevant transition matrix given by the Potts model
and by the corresponding fuzzy Potts model. We further derive conditions for extremality and non-extremality
for the corresponding coarse-grained Markov chain.

We note that non-extremality of the coarse-grained
chain already implies non-extremality of the original chain.
  Indeed, if the coarse-grained chain
is not extremal then there is a tail-measurable event for the coarse-grained
chain which has probability different from zero or one.
But then the preimage of this event under the coarse-graining
map is a tail-measurable event that
has a probability different from zero or one under the original Potts measure,
which therefore must also be not extremal.


 Let us mention that
sequences of coarse-graining maps to increasingly coarse spaces
for which a temporal interpretation is natural have been considered  in \cite{JK}.

Section 4 is related to non-extremality conditions of a TISGM (for the original model),
where we check  the Kesten-Stigum condition (based on the second largest eigenvalue).
We give explicit formulas for critical parameters for the Kesten-Stigum condition to hold. This provides us with
sufficient conditions on the temperature to see non-extremality.
We cannot expect these conditions to be sharp in general, since the Kesten-Stigum bound is not sharp in most cases,
but often the Kesten-Stigum bound is numerically not very far off.

Section 5 deals with extremality conditions for a TISGM.
There are various approaches
in the literature to finding sufficient conditions for extremality which can be reduced to a finite-dimensional
optimization problem based only on the transition matrix,
see the percolation method
of  \cite{MSW},\cite{Mos2}, the symmetric entropy method of \cite{FK}, or for the binary
asymmetric channel (i.e., corresponding to an Ising model with an external field) the readily available bound of  Martin \cite{Mar}. Different techniques
are used  in \cite{BC} to show sharpness of the Kesten-Stigum bound for an Ising channel with very weak asymmetry.

In this paper we employ the approach of
\cite{MSW}. As in a particular temperature region there are many TISGMs which have different transition matrices which all lie
in a sufficiently small neighborhood of the transition matrix of the free measure,
and we know  that extremality holds for the latter, a continuity (of parameters with respect to temperature) argument
provides us with a lower bound on the number of extremal TISGMs which have to occur.

\section{Preliminaries}

{\it Cayley tree.}
The Cayley tree $\Gamma^k$
of order $ k\geq 1 $ is an infinite tree, i.e., a connected graph without
cycles, such that exactly $k+1$ edges originate from each vertex.
Let $\Gamma^k=(V, L)$ where $V$ is the set of vertices and  $L$ the set of edges.
Two vertices $x$ and $y$ are called {\it nearest neighbors} if there exists an
edge $l \in L$ connecting them.
We will use the notation $l=\langle x,y\rangle$.
A collection of distinct nearest neighbor pairs $\langle x,x_1\rangle, \langle x_1,x_2\rangle,...,\langle x_{d-1},y\rangle$ is called a {\it path} from $x$ to $y$. The distance $d(x,y)$ on the Cayley tree is the number of edges of the shortest path from $x$ to $y$.

For a fixed $x^0\in V$, called the root, we set
\begin{equation*}
W_n=\{x\in V\,| \, d(x,x^0)=n\}, \qquad V_n=\bigcup_{m=0}^n W_m
\end{equation*}
and denote
$$
S(x)=\{y\in W_{n+1} :  d(x,y)=1 \}, \ \ x\in W_n, $$ the set  of {\it direct successors} of $x$.

{\it Potts model.} We consider the Potts model on a Cayley tree,
where to each vertex of the tree
a spin variable is assigned which takes values in the
local state space
$\Phi:=\{1,2,\dots,q\}$.  For $A\subseteq V$ a spin {\it configuration}
$\sigma_A$ on $A$ is defined as a function $$x\in
A\to\sigma_A(x)\in\Phi.$$

The set of all configurations coincides with $\Omega_A=\Phi^{A}$. We denote $\Omega=\Omega_V$ and
$\sigma=\sigma_V.$

The (formal) Hamiltonian of the Potts model is
\begin{equation}\label{ph}
H(\sigma)=-J\sum_{\langle x,y\rangle\in L}
\delta_{\sigma(x)\sigma(y)},
\end{equation}
where $J\in R$ is a coupling constant,
$\langle x,y\rangle$ stands for nearest neighbor vertices and $\delta_{ij}$ is the Kroneker's
symbol:
$$\delta_{ij}=\left\{\begin{array}{ll}
0, \ \ \mbox{if} \ \ i\ne j\\[2mm]
1, \ \ \mbox{if} \ \ i= j.
\end{array}\right.
$$
In the present paper we only consider the case of ferromagnetic interaction $J>0$.

For a finite domain $D\subset V$ with the boundary condition $\varphi_{D^c}$ given on its
complement $D^c=V\setminus D,$ the conditional Hamiltonian is
\begin{equation}H(\sigma_D\big
| \varphi_{D^c})=-J\sum_{{\langle x,y\rangle,\atop x,y\in D}}
\delta_{\sigma(x)\sigma(y)}-J\sum_{{\langle x,y\rangle,\atop x\in D,\, y\in D^c}}
\delta_{\sigma(x)\varphi(y)}.
\end{equation}

{\it Gibbs measure.}
 A probability measure $\mu$ on $(\Omega,{\mathcal B})$ (where ${\mathcal B}$ is the
$\sigma$-algebra generated by cylinder subsets of $\Omega$) is called a {\it Gibbs measure} (with
Hamiltonian $H$) if it satisfies the Dobrushin-Lanford-Ruelle (DLR) equation:
for all finite $D\subset V$ and
$\sigma_D\in\Omega_D$:
\begin{equation}
\mu\left(\left\{\omega\in\Omega:\;
\omega\big|_{D}=\sigma_D\right\}\right)= \int_{\Omega}\mu ({\rm d}\varphi)\nu^{D}_\varphi
(\sigma_D),
\end{equation}
 where $\nu^{D}_{\varphi}$ is the conditional probability (finite-volume Gibbs measure):
\begin{equation}
 \nu^{D}_{\varphi}(\sigma_D)=\frac{1}{Z_{D,\varphi}}\exp\;\left(-\beta H
\left(\sigma_D\big |\,\varphi_{D^c}\right)\right).
\end{equation}
 Here $\beta={1\over T}$, where
 $T>0$ is the temperature and $Z_{D, \varphi}$ stands for the partition function in $D$, with the boundary
condition $\varphi$:
$$Z_{D, \varphi}=
\sum_{{\widetilde\sigma}_D\in\Omega_{D}} \exp\;\left(-\beta H \left({\widetilde\sigma}_D\,\big
|\,\varphi_{D^c}\right)\right).$$

{\it Splitting Gibbs measure.} For $n\geq 1$ we denote by $\sigma_n$ a
configuration $\sigma_n:V_n\to \Phi$, i.e.,  $\sigma_n\equiv \sigma_{V_n}=\{\sigma(x), x\in V_n\}$.

Define a distribution in the volume $V_n$ as
\begin{equation}\label{p*}
\mu_n(\sigma_n)=Z_n^{-1}\exp\left\{-\beta H_n(\sigma_n)+\sum_{x\in W_n}{\tilde h}_{\sigma(x),x}\right\},
\end{equation}
where $Z_n^{-1}$ is the normalizing factor, $\{{\tilde h}_x=({\tilde h}_{1,x},\dots, {\tilde h}_{q,x})\in R^q, x\in V\}$ is a collection of vectors and
$H_n(\sigma_n)$ is the restriction of the Hamiltonian to $V_n$, i.e.
$$H_n(\sigma_n)=-J\sum_{\langle x,y\rangle\in V_n}
\delta_{\sigma(x)\sigma(y)}.$$

We say that the probability distributions (\ref{p*}) are compatible if for all
$n\geq 1$ and $\sigma_{n-1}\in \Phi^{V_{n-1}}$:
\begin{equation}\label{p**}
\sum_{\omega_n\in \Phi^{W_n}}\mu_n(\sigma_{n-1}\vee \omega_n)=\mu_{n-1}(\sigma_{n-1}).
\end{equation}
Here $\sigma_{n-1}\vee \omega_n$ is the concatenation of the configurations $\sigma_{n-1}$ and $\omega_n$.
In this case, by Kolmogorov's well-known extension theorem\footnote{The equality (\ref{p**}) means that
the restriction of $\mu_n$ to the set $\Phi^{V_{n-1}}$ coincides with $\mu_{n-1}$.}, there exists a unique measure $\mu$ on $\Phi^V$ such that,
for all $n$ and $\sigma_n\in \Phi^{V_n}$,
\begin{equation}\label{dlr}
\mu\left(\left\{\omega\in\Omega:\; \omega\big|_{V_n}=\sigma_n\right\}\right)=\mu_n(\sigma_n).
\end{equation}

Such a measure is called a {\it splitting Gibbs measure}\footnote{
To see that a SGM satisfies the DLR equation, we consider any finite volume $D$
and note that for any finite $n$ which is sufficiently large we have
\begin{equation}
\mu_n\left(\left\{\omega\in\Omega:\;
\omega\big|_{D}=\sigma_D\right\}\right)= \int_{\Omega}\mu_n ({\rm d}\varphi)\nu^{D}_\varphi
(\sigma_D),
\end{equation}
which follows from the compatibility property of the finite-volume Gibbs measures.
} (SGM) corresponding to the Hamiltonian (\ref{ph}) and the vector-valued function ${\tilde h}_x, x\in V$.

The following statement describes conditions on ${\tilde h}_x$ guaranteeing compatibility of $\mu_n(\sigma_n)$.

\begin{thm}\label{ep} \label{Theorem1} (see \cite{Ga8}, \cite[p.106]{Ro}) The probability distributions
$\mu_n(\sigma_n)$, $n=1,2,\ldots$, in
(\ref{p*}) are compatible iff for any $x\in V\setminus\{x^0\}$
the following equation holds:
\begin{equation}\label{p***}
 h_x=\sum_{y\in S(x)}F(h_y,\theta),
\end{equation}
where $F: h=(h_1, \dots,h_{q-1})\in R^{q-1}\to F(h,\theta)=(F_1,\dots,F_{q-1})\in R^{q-1}$ is defined as
$$F_i=\ln\left({(\theta-1)e^{h_i}+\sum_{j=1}^{q-1}e^{h_j}+1\over \theta+ \sum_{j=1}^{q-1}e^{h_j}}\right),$$
$\theta=\exp(J\beta)$, $S(x)$ is the set of direct successors of $x$ and $h_x=\left(h_{1,x},\dots,h_{q-1,x}\right)$ with
\begin{equation}\label{hh}
h_{i,x}={\tilde h}_{i,x}-{\tilde h}_{q,x}, \ \ i=1,\dots,q-1.
\end{equation}
\end{thm}

From Theorem \ref{ep} it follows that for any $h=\{h_x,\ \ x\in V\}$
satisfying (\ref{p***}) there exists a unique SGM $\mu$ for the Potts model.

To compare with the literature we remark that the quantities $\exp(\tilde h_{i,x})$ define a boundary
law in the sense of Definition 12.10 in Georgii's book \cite{Ge}. Here one should take the boundary
law $z_{x,y}(i)=\exp(\tilde h_{i,x})$, where $y$ is the unique neighbor of $x$ that lies closer to $x^0$.
Now all the other $z_{x,y}$'s can be
calculated inductively from these using the definition of a boundary law,
and this always yields a unique and well-defined boundary law.

The equation (\ref{p***}) can be written as the equation (12.15) in \cite{Ge} for $z_{x,y}(i)$ normalized at $q$,
i.e, assuming $z_{x,y}(q)\equiv 1$. Note that for a given $h_x=(h_{1,x},\dots,h_{q-1,x})$, assuming
that $\tilde{h}_{q,x}$ is also some given number
one can obtain $\tilde h_{i,x}$, $i=1,\dots,q-1$ from (\ref{hh}). Here the normalization of the
boundary law mentioned above corresponds to the $\tilde{h}_{q,x}\equiv 0$.
Thus by the normalization we reduce the $q$ dimensional vector to the $q-1$ dimensional vector.
One can normalize the $q$ dimensional vector by the normalization at any other coordinate
instead of $q$th coordinate. In this paper for definiteness we only consider the normalization at $q$.

Compare also Theorem 12.2 in \cite{Ge} which describes the connection between boundary laws and finite-volume marginals of
splitting Gibbs measures for general spin models.
Formula (\ref{p*}) is a special case of this, since it was
claimed only for volumes $V_n$ which were assumed to
be balls on the graph. However,
from the defining property of boundary laws as solutions to the fixed point equation
it follows immediately that \eqref{p*} extends to general subtrees $V_n$ of any shape,
where $W_n$ is the inner boundary.
For marginals on the two-site volumes which consist of two adjacent sites $x,y$,
assuming a translation-invariant
boundary law $z$ we have in the case of the Potts model
$$\mu(\s_x=i,\s_y=j)= \frac{1}{Z} z_i \exp(J\beta\delta_{i j}) z_j.$$
From this  the relation between the boundary law and the transition matrix
for the associated tree-indexed Markov chain\footnote{We recall that a tree-indexed Markov chain is defined as follows. Suppose we are given a tree $\mathcal T$ with vertex set $V$, a probability measure $\nu$, and a transition matrix $\mathbb{P}=(P_{ij})_{i,j\in \Phi}$ on the finite set $\Phi=\{1,\dots,q\}$. We can obtain a tree indexed Markov chain $X : V \to \Phi$ by choosing $X(x^0)$ according to $\nu$ and choosing $X(v)$, for each vertex $v\ne x^0$, using the transition probabilities given the value of its parent, independently of everything else. See Definition 12.2 in \cite{Ge} for a detailed definition. In our case the matrix $\mathbb{P}$ is given by formula (\ref{epe2})} (splitting Gibbs measure) is immediately obtained
from the formula of the conditional probability. Indeed, the
corresponding transition matrix giving the probability to go from a state $i$
at site $x$ to a state
to $j$ at the neighbor $y$
is then
\begin{equation}\label{epe2}
P_{ij}={z_j\exp(J\beta\delta_{ij})\over \sum_{r=1}^qz_r\exp(J\beta\delta_{ir})}.
\end{equation}

In the language of boundary fields:
Once we have $\tilde{h}_x$ which is independent of $x$, i.e. $\tilde{h}=(\tilde{h}_1,\dots,\tilde{h}_q)$, then the tree-indexed Markov chain with states
 $\{1,2,\dots,q\}$ is given by the transition matrix $\mathbb P=(P_{ij})$ with
 $$P_{ij}={\exp\left(J\beta\delta_{ij}+\tilde{h}_j\right)\over \sum_{r=1}^q\exp\left(J\beta\delta_{ir}+\tilde{h}_r\right)}.$$

{\it Translation-invariant SGMs.} A translation-invariant splitting Gibbs measure (TISGM) corresponds to a solution $h_x$ of (\ref{p***}) with  $h_x=h=(h_1,\dots,h_{q-1})\in R^{q-1}$ for all $x\in V$. Then from equation (\ref{p***}) we get $h=kF(h,\theta)$, and denoting $z_i=\exp(h_i), i=1,\dots,q-1$, the last equation can be written as
\begin{equation}\label{pt1}
z_i=\left({(\theta-1)z_i+\sum_{j=1}^{q-1}z_j+1\over \theta+ \sum_{j=1}^{q-1}z_j}\right)^k,\ \ i=1,\dots,q-1.
\end{equation}

In \cite{KRK} all solutions of the equation (\ref{pt1}) are given. By these solutions the full set of TISGMs is described. In particular, the following results are obtained which will be the starting point of
the present analysis and which we repeat for convenience of the reader.

\begin{thm}\label{Theorem2}\cite{KRK} For any solution $z=(z_1,\dots,z_{q-1})$ of the system of equations (\ref{pt1}) there exist $M\subset \{1,\dots,q-1\}$ and $z^*>0$ such that
$$z_i=\left\{\begin{array}{ll}
1, \ \ \mbox{if} \ \ i\notin M\\[3mm]
z^*, \ \ \mbox{if} \ \ i\in M.
\end{array}
\right.
$$
\end{thm}

Thus any TISGM of the Potts model corresponds to a solution of the following equation
\begin{equation}\label{rm}
z=f_m(z)\equiv \left({(\theta+m-1)z+q-m\over mz+q-m-1+\theta}\right)^k,
\end{equation}
for some $m=1,\dots,q-1$.
\begin{rk} We note that for each fixed $m$, the equation (\ref{rm}) has up to
three solutions: $z_0=1, z_i=z_i(\theta,q,m), i=1,2$, with $z_1<z_2$
(see \cite[Step 1 of the proof of Theorem 1]{KRK}).
\end{rk}
Denote
\begin{equation}\label{tm}
\theta_m=1+2\sqrt{m(q-m)}, \ \ m=1,\dots,q-1.
\end{equation}
It is easy to see that
\begin{equation}\label{st}
\theta_m=\theta_{q-m} \ \ \mbox{and} \ \ \theta_1<\theta_2<\dots<\theta_{\lfloor{q\over 2}\rfloor-1}<\theta_{\lfloor{q\over 2}\rfloor}\leq q+1.
\end{equation}

\begin{pro}\label{pw}\cite{KRK} Let $k=2$, $J>0$.
\begin{itemize}
\item[1.]
If $\theta<\theta_1$ then there exists a unique TISGM;
\item[2.]
If $\theta_{m}<\theta<\theta_{m+1}$ for some $m=1,\dots,\lfloor{q\over 2}\rfloor-1$ then there
are  $1+2\sum_{s=1}^m{q\choose s}$ TISGMs which correspond (by Theorem \ref{ep}) to the
solutions $z_r\equiv z_r(\theta, q, s)=x^2_r(s,\theta)$, $r=1,2$, $s=1,\dots,m$ of (\ref{rm}), for $m=s$, i.e. $f_s(z)=z$, where
\begin{equation}\label{s}\begin{array}{ll}
x_{1}(s,\theta)={\theta-1-\sqrt{(\theta-1)^2-4s(q-s)}\over 2s},\\[2mm] x_{2}(s,\theta)={\theta-1+\sqrt{(\theta-1)^2-4s(q-s)}\over 2s}.
\end{array}
\end{equation}
\item[3.] If $\theta_{\lfloor{q\over 2}\rfloor}<\theta\ne q+1$ then there are $2^q-1$ TISGMs;
\item[4] If $\theta=q+1$ the number of TISGMs is as follows
$$\left\{\begin{array}{ll}
2^{q-1}, \ \ \mbox{if} \ \ q \ \ \mbox{is odd}\\[2mm]
2^{q-1}-{q-1\choose q/2}, \ \ \mbox{if} \ \ q \ \ \mbox{is even;}
\end{array}\right.$$
\item[5.] If $\theta=\theta_m$, $m=1,\dots,\lfloor{q\over 2}\rfloor$, \,($\theta_{\lfloor{q\over 2}\rfloor}\ne q+1$) then number of TISGMs is
$$1+{q\choose m}+2\sum_{s=1}^{m-1}{q\choose s}.$$
\end{itemize}
\end{pro}
\begin{rk} We note that constants $\theta_1$, $\dots$, $\theta_{\lfloor{q\over 2}\rfloor}$
with similar properties as in Proposition \ref{pw} can be defined for general $k$,
 see \cite[Theorem 1]{KRK}.
However, for $k\geq 3$ Theorem 1 of \cite{KRK} gives only the existence of the corresponding
solutions to the equation (\ref{rm}).  To investigate the (non-)extremality for a given TISGM
in this paper we need
explicit formulas of these solutions, which we have {\rm only} for $k=2$ by Proposition \ref{pw}.
\end{rk}
\begin{rk} For each fixed $m\leq \lfloor{q\over 2}\rfloor$, the symmetry of the
Potts model allows us to reduce the study to at most three TISGMs: $\mu_0$ which corresponds to the
solution $h_x=(0,\dots,0)$ and two TISGMs $\mu_1(\theta,m)$, $\mu_2(\theta,m)$ which correspond to the vectors
$$h_x=h_1=(\underbrace{\ln z_1,\ln z_1, \dots,\ln z_1}_m,0,\dots,0),$$
$$h_x=h_2=(\underbrace{\ln z_2,\ln z_2, \dots,\ln z_2}_m,0,\dots,0),$$
where $z_1=z_1(\theta,m)$ and $z_2=z_2(\theta,m)$ are solutions of the equation (\ref{rm}).
Hence, for each given $m$
all TISGMs belong to three classes: the first class contains only $\mu_0$, the second class all measures
which correspond to vectors obtained by permutations of coordinates of $h_1$, the last class contains all TISGMs
which correspond to a vector obtained by  permutations of $h_2$. It should also be noted that in case $m=q/2$ the solutions
$h_1$ and $h_2$ define the same TISGMs after a relabeling of indices and a renormalization of $h_1$.
\end{rk}
\section{Fuzzy transformation to Ising model with an external field}
In order to study extremality of a TISGM corresponding to a solution $z>0$, $z\ne 1$ of (\ref{rm}), (i.e. to a vector $l=(l_1,\dots,l_q)\in R^q$ with $m$ coordinates equal to $z$ and $q-m$ coordinates equal to 1) we divide the coordinates of this vector into two classes: We write
$l'$ and $l''$, where $l'$ has $m$ coordinates each equal to $z$ and $l''$ has $q-m$ coordinates each equal to 1.

This is a particular case of the well-known fuzzy Potts model.
The fuzzy Potts model on a given graph is obtained from a standard
$q$-valued Potts model on the same graph
by regrouping the possible spin-values into $r$ classes with $r<q$.
The use of the word 'fuzzy' suggests that only partial information of the underlying model is kept.
To our knowledge this model has appeared
for the first time as a lattice model in \cite{MV}.

For our purposes we will not consider these general fuzzy transformations, but only
transformations which are adapted
to the symmetry of the Potts tree measures we have constructed.  By this we mean
that only those colors which behave similarly under the Potts tree measure will be put into the same class, as we will explain now.

Without loss of generality, by relabeling of coordinates,  we can take $l$, $l'$ and $l''$ as follows:
$$l=(\underbrace{z,z, \dots,z}_m,\underbrace{1,1,\dots,1}_{q-m})=
((\underbrace{z,z, \dots,z}_m),(\underbrace{1,1,\dots,1}_{q-m}))=(l',l'').$$

Define a (fuzzy) map $T:\{1,2,\dots,q\}\to \{1',2'\}$ as
\begin{equation}\label{T}
T(i)=\left\{\begin{array}{ll}
1', \ \ \mbox{if} \ \ i\leq m\\[2mm]
2', \ \ \mbox{if} \ \ i\geq m+1.
\end{array}
\right.
\end{equation}

This map identifies spin-values which have the same value of the boundary law and
are treated in an equal fashion by the transition matrix. We extend this map to act on infinite-volume spin configurations
and measures in the infinite volume, i.e.
$$T(\sigma)(x)=T(\sigma(x)), \ \ T(\mu)(A)=\mu(T^{-1}(A)).$$

We note that a TISGM corresponding to a vector $l\in R^q$ is a tree-indexed Markov chain
with states $\{1,2,\dots,q\}$ and transition probabilities matrix $\mathbb P=(P_{ij})$ with $P_{ij}$ given by (\ref{epe2}).

From (\ref{epe2}) we get
\begin{equation}\label{P}
P_{ij}=\left\{\begin{array}{llllll}
\theta z/Z_1, \ \ \mbox{if} \ \ i=j, \, i\in \{1,\dots,m\}\\[2mm]
z/Z_1, \ \ \mbox{if} \ \ i\ne j, \, i,j\in \{1,\dots,m\}\\[2mm]
1/Z_1, \ \ \mbox{if} \ \ i\in \{1,\dots,m\}, j\in\{m+1,\dots,q\}\\[2mm]
z/Z_2, \ \ \mbox{if} \ \ i\in\{m+1,\dots,q\}, j\in \{1,\dots,m\}\\[2mm]
\theta/Z_2, \ \ \mbox{if} \ \ i=j, \, i\in \{m+1,\dots,q\}\\[2mm]
1/Z_2, \ \ \mbox{if} \ \ i\ne j, \, i,j\in \{m+1,\dots,q\},
\end{array}
\right.
\end{equation}
where
$$Z_1=(\theta+m-1)z+q-m, \ \ \ Z_2=mz+\theta+q-m-1.$$

The fuzzy map $T$ reduces the matrix $ \mathbb P$ to the $2\times 2$ matrix $T(\mathbb P)=\hat{\mathbb P}=(p_{ij})_{i,j=1',2'}$ with
   \begin{equation}\label{TP}
p_{ij}=\left\{\begin{array}{llll}
{(\theta+m-1)z\over Z_1}, \ \ \mbox{if} \ \ i=j=1'\\[2mm]
{q-m\over Z_1}, \ \ \ \ \ \ \mbox{if} \ \ i=1', j=2'\\[2mm]
 {mz\over Z_2}, \ \ \ \ \ \ \mbox{if} \ \ i=2', j=1'\\[2mm]
{\theta+q-m-1\over Z_2}, \ \ \mbox{if} \ \ i=j=2'.
\end{array}
\right.
\end{equation}
More precisely this means: Consider the translation-invariant tree-indexed Markov chain $\mu$
with transition matrix given by (\ref{P}).
Then its image measure $T(\mu)$ under the
site-wise application of the fuzzy map
$T$ is a tree-indexed Markov chain
with local state space $1',2'$ with
the transition matrix given by (\ref{TP}).

 Note that an application to a Gibbs measure which
is a Markov chain, of a transformation which is
not adapted to the structure of the transition matrix
(for example a fuzzy map which identifies spin values which have different values of the boundary law)
would in general not give rise to a Markov chain (and possibly not even to a Gibbs measure with
an absolutely summable interaction potential).

We also observe the following fact:  If a Markov chain $\mu$ is extremal in the set of
Gibbs-measures for the Potts model then this implies its triviality on the
tail-sigma algebra of the $q$-spin events (see \cite[Theorem 7.7]{Ge}). Indeed,
the extremality of Gibbs measures on any graph is equivalently characterized by the triviality
of the measure on the tail-sigma algebra.
But this implies that  the mapped chain $T(\mu)$ is trivial on its
tail-sigma algebra (since the latter can be identified
with the sub-sigma algebra of events in the tail-sigma algebra of the $q$-spin events which
do not distinguish spin-values which have the same values of the fuzzy variable).

It turns out that we can lift the action of the coarse graining to obtain an effective Hamiltonian
for the coarse-grained Markov chain which has the interpretation of an Ising model in a
magnetic field which is vanishing if and only if $2m =q$.
Note that this Hamiltonian will be $z$-independent.

Since $z$ and $m$ define a boundary law $l$,
when we say boundary law $z$ we mean the boundary law defined by $z$ for fixed $m$.

We have the following proposition.

\begin{pro} Let $k\geq 2$ be an integer. Consider a translation invariant Markov chain
$\mu_{\theta, m,z}$ for the Potts model with coupling constant parameterized by $\theta$
(see (\ref{P})), for $m\leq q/2$ and a corresponding (to $\mu_{\theta, m,z}$) choice
of the (up to three) values of the boundary law $z$ (solution of (\ref{rm})).

Then there exist a coupling constant $J'=J'(\theta,m)$ and a magnetic field value $h'=h'(\theta,m)$, where both values are
independent of the boundary law $z$,
such that the coarse-grained measure $T(\mu_{\theta, m,z})$
is a Gibbs measure  for the Hamiltonian
\begin{equation}\label{H'}
H'(\phi)=- J'\sum_{\langle x, y\rangle }\phi(x)\phi(y)- h'\sum_x \phi(x)
\end{equation}
of the corresponding Ising model with spin variables
$\phi(x)=\pm 1$  on the Cayley tree of order $k$. The measure $T(\mu_{\theta, m,z})$
 is a tree-indexed Markov chain (TISGM) which has a matrix $M$
as its transition matrix. Here we have identified the fuzzy class $1'$ with the Ising spin value $+1$
and the fuzzy class $2'$ with the Ising spin value $-1$.

For the coupling constant we have
\begin{equation}\begin{split}\label{transe13}
e^{4 J'}=\frac{(\theta + m-1)(\theta + q -m -1)}{(q-m)m}
\end{split}
\end{equation}	
For the magnetic field we have
\begin{equation}\begin{split}\label{transe14}
e^{\frac{4 h'}{k+1}}=\frac{(\theta + m-1)m}{(q-m)(\theta + q -m -1)}(\frac{z}{s})^2
\end{split}
\end{equation}	
where $z$ denotes the value of the boundary law for the Potts model and
$(s,1)$ denotes the value of the corresponding boundary law for the Ising model.

The boundary laws for the Potts model and the corresponding Ising model satisfy the relation
\begin{equation}\begin{split}\label{transe15}
s=(\frac{m}{q-m})^{\frac{k}{k+1}}z
\end{split}
\end{equation}	
which is independent of the choice of the solution at fixed $m$ and makes the magnetic field $h'$ independent
of the choice of $z$ at fixed $m$.

Moreover, the Hamiltonian (\ref{H'}) does not have any TISGM that is not of the form $T(\mu_{\theta, m,z})$.
\end{pro}

It is interesting to compare this result with Proposition 4.1 in
\cite{HaKu04} which states that any fuzzy image of the free b.c. condition Potts measure (which is obtained
for $z=1$) is quasilocal. Here we extend this result to the larger class of Markov chains with fixed $m$
and  remark that the Hamiltonian of the fuzzy model with classes $m$ and $q-m$
stays the same when we take a different boundary law $z$.
Hence it suffices to look at the free measure to construct this Hamiltonian.

\begin{proof}
The transfer matrix of the Ising model whose Hamiltonian is to be constructed
has the form

\begin{equation}\begin{split}\label{transe}
\hat Q =\Bigl( e^{J' \phi(x)\phi(y)+\frac{h'}{k+1}\phi(x)+ \frac{h'}{k+1}\phi(y)}\Bigr)_{\phi(x)=\pm 1, \phi(y)=\pm 1}
=\begin{pmatrix} e^{J' + \frac{2 h'}{k+1}} &  e^{-J' } \cr
 e^{-J' }  & e^{J' - \frac{2 h'}{k+1}} \cr
\end{pmatrix}
\end{split}
\end{equation}
(compare e.g. (12.20) of Georgii's book). The corresponding equation
for a boundary law $(s_x,1)$ (which we allow to be $x$-dependent at this stage) for the Ising model is written as
\begin{equation}\begin{split}\label{transe17}
s_x=\prod_{y\in S(x)}
\frac{e^{J' + \frac{2 h'}{k+1}} s_y+  e^{-J' } }{e^{-J' }s_y  + e^{J' - \frac{2 h'}{k+1}}}.
\end{split}
\end{equation}
Suppose we have a homogeneous solution $s$ of this equation for the boundary law of the Ising model. Then the corresponding transition matrix is
\begin{equation}\begin{split}\label{transe}
\hat P_s =\begin{pmatrix} \frac{e^{J' + \frac{2 h'}{k+1}} s}{e^{J' + \frac{2 h'}{k+1}} s+  e^{-J' }  } &  \frac{e^{-J' }}{e^{J' + \frac{2 h'}{k+1}} s+  e^{-J' }  }  \cr\\[2mm]
\frac{ e^{-J' }s}{  e^{-J' }s + e^{J' - \frac{2 h'}{k+1}}} & \frac{e^{J' - \frac{2 h'}{k+1}}}{   e^{-J' }s + e^{J' - \frac{2 h'}{k+1}}}\cr
\end{pmatrix}
\end{split}.
\end{equation}
Recall the form of  $M=\hat P$ by which we denote the stochastic $2\times 2$-transition matrix
we obtained from the application of the fuzzy map $T$ to an $m$- and $z$-dependent
splitting Gibbs measure $\mu$ for the Potts model as described above.

Equating this transition matrix with the transition matrix obtained from the coarse-graining of the boundary
laws for the Potts model for given $m$ and $z$ we find
\begin{equation}\begin{split}\label{transe-19}
e^{2 J' + \frac{2 h'}{k+1}}s =\frac{ p_{11}}{p_{12}}= \frac{\theta + m -1}{q-m }
\end{split}
\end{equation}	
and
\begin{equation}\begin{split}\label{transe-20}
e^{- 2 J' + \frac{2 h'}{k+1}}s =\frac{ p_{21}}{p_{22}}= \frac{m}{\theta + q -m -1 }z
\end{split}.
\end{equation}	
Taking quotient (respectively product) of these equations the formulas for coupling and magnetic field
\eqref{transe13} \eqref{transe14} follow.
To see the relation \eqref{transe15}  between homogeneous boundary laws of the Potts model
and the Ising model we start with the homogeneous version of  the Ising boundary law equation \eqref{transe17} to obtain
\begin{equation}\begin{split}\label{transe466}
s&=s^{-k }\Bigl(\frac{e^{2 J' + \frac{2 h'}{k+1}} s+  1 }{1 + ( s e^{- 2 J' + \frac{2 h'}{k+1}})^{-1}}\Bigr)^k\cr
&=\Bigl(\frac{m}{q-m}\frac{z}{s} \Bigr)^k  \left({(\theta+m-1)z+q-m\over mz+q-m-1+\theta}\right)^k \cr
&=\Bigl(\frac{m}{q-m}\frac{z}{s} \Bigr)^k  z,
\end{split}
\end{equation}
where we have substituted in the second line \eqref{transe-19} and \eqref{transe-20} to bring the Potts parameters into play
and  recognized the r.h.s. of the equation for the boundary
law of the Potts model $z=f_m(z)$ to get the third line. This proves the relation between the boundary laws \eqref{transe15} independently
of the choice of the solution for $m$ fixed (and independently of temperature of the Potts model).

Let us go back to \eqref{transe17} which is the functional (spatially dependent) equation for the boundary law for the Ising model at fixed $J',h'$.
The relation between boundary laws also works for spatially dependent boundary laws and we get that
$s_x $ is a solution for the Ising model \eqref{transe17}  if and only if
the quantity
$z_x:=s_x \Bigl( \frac{q-m}{m} \Bigr)^\frac{k}{k+1}$ is a solution
 of the functional equation
\begin{equation}\label{Tz1}
z_x=\prod_{y\in S(x)} {(\theta+m-1)z_y+q-m\over mz_y+q-m-1+\theta}.
\end{equation}
The proof of this statement uses the same substitutions as in the homogeneous case.

To give all possible TISGMs of the Hamiltonian (\ref{H'}) one has to solve the equation (\ref{Tz1}) in class of constant functions, $z_x=z$. Then this equation coincides with the equation (\ref{rm}), thus only TISGMs of the form $T(\mu_{\theta,m,z})$ exist.
\end{proof}

\bigskip

We come to a structural result which gives another way of looking
at the Potts measures, in particular in the case of boundary laws which are
different from $1$ in terms of simpler measures.

\begin{thm}\label{Constructive}
Consider the following decomposition of
the Potts measures for class size $m$ and one of the
 values of the boundary law $z$
as an integral over the image measure under the fuzzy map $T$, and the conditional measure
$$\mu_{\theta,m,z}(d\s)=\int (T\mu_{\theta,m,z})(d\s')\mu_{\theta,m,z}(d\s| T=\s').
$$ 

Let us denote, for given fuzzy configuration $\s'$ in
the infinite volume, the connected components of the level sets $\{v\in V| \s'(v)=1'\}$ and $\{v\in V| \s'(v)=2'\}$ by $C_i$ and $G_j$ respectively.

Then the conditional measure factorizes,
\begin{equation}\label{322}
\mu_{\theta,m,z}(d\s| T=\s')=\prod_i \mu_{\theta,m,z}(d\s_{C_i}| T=\s') \prod_j \mu_{\theta,m,z}(d\s_{G_j}| T=\s'),
\end{equation}
where the conditional measures on the fuzzy-spin clusters $C_i$ corresponding to
the fuzzy-class $1'$
take the simple form
\begin{equation}\label{323}\mu_{\theta,m,z}(d\s_{C_i}| T=\s')=\mu^\text{free}_{\{1,\dots,m\},\theta,C_i}(d\s_{C_i}).
\end{equation} 
Here the measure on the r.h.s. is a free boundary condition measure for a Potts model
with $m$ spins with possible values $\{1,\dots,m\}$ on the subtree $C_i$.

Similarly, the conditional measure on the fuzzy-spin cluster $G_j$ corresponding to
fuzzy-class $2'$ is a free boundary condition Potts measure with states $\{m+1,\dots,q\}$, at
parameter $\theta$,  and we have
\begin{equation}\label{324}\mu_{\theta,m,z}(d\s_{G_j}| T=\s')=\mu^\text{free}_{\{m+1,\dots,q\},\theta,G_j}(d\s_{G_j}).
\end{equation} 

\end{thm}

\begin{rk} Note that the level sets $C_i$ resp. $G_j$ are not required
to be finite.

The theorem helps us to get a more intuitive understanding of
the Potts measures $\mu_{\theta,m,z}$,
for class size $m$ and the various corresponding values of the boundary law $z\neq 1$,
in terms of well-known measures obtained from standard boundary conditions.

We already understood that $T\mu_{\theta,m,z}$, for $z\neq 1$, is equal to either
the maximal measure corresponding to the Ising-Hamiltonian $H'$, or the minimal measure
corresponding to $H'$.
By known properties of the Ising model in an external magnetic field (which will in general
be non-zero)
these two fuzzy measures can also be obtained as limits of finite-volume Ising measures
with homogeneous boundary conditions either equal to $1'$ for all sites, or equal to $2'$ for all sites.
Hence we can obtain the construction of
the measure $\mu_{\theta,m,z}$ in terms of the following two steps:

Step A) Given $m$ we first compute the $m$-dependent coarse-grained
Hamiltonian $\eqref{H'}$.
We choose boundary condition $1'$  (or $2'$) to construct first the fuzzy measure, which is identical
to the well-known minimal (or maximal) Ising measure.

Step B) Given the construction of the coarse-grained measure from
step A) we put free states for the Potts model
with coupling parameter $\theta$ on each of the connected
components of the set of sites of $\s'\equiv 1'$ and of $\s'\equiv 2'$.
\end{rk}

\begin{proof}
We verify the desired representation of the conditional measure
on all local events of the form $\{\s_V=\zeta_V\}$ for a fixed
finite subtree $V$ containing the origin.
It suffices to assume that $T(\zeta_V)=\s'_V$, otherwise both sides of \eqref{322} are zero.
To reduce the computations
 to finite-volume expressions, we note that in general,
by martingale convergence,
\begin{equation}
\begin{split}
\mu_{\theta,m,z}(\s_V=\zeta_V| T=\s')=\lim_{W}\mu_{\theta,m,z}(\s_V=\zeta_V| T(\s_W)=\s'_W)
\end{split}
\end{equation}
along any sequence of subtrees $W$ converging to the full vertex set of the tree, for
$T\mu_{\theta,m,z}$-a.e- fuzzy configuration in the infinite volume $\s'$.

For $V$ fixed and $W$ large enough but finite,
we use the elementary formula for conditional probabilities to write the expression
$\mu_{\theta,m,z}(\s_V=\zeta_V| T(\s_W)=\s'_W)$
in terms of finite sums over products of
elements of the transition matrix of the Markov chain
$\mu_{\theta,m,z}$. Using cancellations between numerator and denominator,
 we see in this way that it becomes independent of $W$.
To see that the factorization \eqref{322} takes place over the fuzzy-clusters (where we have
to take into account only of those clusters which intersect the finite volume $V$)
we note that the probabilities $P_{ij}$ defined by (\ref{P}) for $i$ and $j$ which are not in the same class depend only on $T(i)$ and $T(j)$, but not on $i$ and $j$ themselves.

Finally, to see why the free measure $\mu^\text{free}_{\{1,\dots,m\},\theta,C_i}$ appears in \eqref{323}
let us fix one component $C_i$ (which may be infinite or not) and look at the probability
of the event $\{\s_{V\cap C_i}= \zeta_{V\cap C_i}\}$  in the measure $\mu_{\theta,m,z}(d \s| T=\s')$. The crucial point of the argument is the following.
Each of the entries of the $m\times m$ subblock
in the corresponding $z$-dependent full $q\times q$ transition matrix carries a uniform factor $z$
which can be pulled out. In this way one obtains that the conditional measure on $C_i$
equals a free measure ($z=1$-measure) for an $m$-state Potts model with
states $\{1,\dots,m\}$ on the
(in general inhomogeneous) tree $C_i$. A similar argument works
for the components $G_j$.
 This proves the theorem.
\end{proof}

\bigskip

 Let us continue by presenting relations between the transition matrices
	of the Potts model and the mapped model which will be very useful below.
As the explicit formulas of eigenvalues are very important
in our checking of (non-)extremality, let us calculate
these eigenvalues now.

	We note that the matrix $\hat{\mathbb P}$ has the two eigenvalues $1$ and
\begin{equation}\label{lambda2}
\lambda_2(\hat{\mathbb P})={(\theta+m-1)z\over Z_1}+ {\theta+q-m-1\over Z_2}-1.
\end{equation}
Let us comment on some general properties of transition matrices which have the same type of two-block symmetry as our transition matrix has:
Our matrix $\mathbb P$ is a stochastic $q\times q$-matrix,
in block-form which has as parameters the block size $m$ and can be written as follows in terms of
four real parameters $p_1$, $p_2$, $q_1$, $q_2$\footnote{The simple expressions of the parameters in terms of existing quantities $z$, $\theta$, $q$, and $m$ are given in the proof of Corollary \ref{ca} below.}.

Here
$p_1$ is the transition rate for going from a state in $1'$ to a different state in $1'$,

$q_1$ is the transition rate for going from a state in $2'$ to a different state in $2'$,

 $p_2$ is the transition rate for going from a state in $1'$ to a state in $2'$,

$q_2$ is the transition rate for going from a state in $2'$ to a state in $1'$.

The form is
\begin{equation}\begin{split}\label{transe}
\mathbb P =(P_{i,j})_{i,j=1,\dots,q}=
\begin{pmatrix} a E_m & 0\cr
0 & b E_{q-m}\cr
\end{pmatrix} +
\begin{pmatrix} p_1 I_{m,m} & p_2 I_{m,q-m}\cr
q_2 I_{q-m,m} & q_1 I_{q-m,q-m}\cr
\end{pmatrix}
\end{split}
\end{equation}
where $I_{k,l}$ is the $k\times l$ matrix which has all matrix elements equal to $1$
and $E_k$ is the $k$-dimensional unity matrix.
Here necessarily $a=1-m p_1-(q-m)p_2$ and $b=1-m q_2-(q-m)q_1$, so that the matrix is stochastic.

We consider the action of the fuzzy map on probability vectors and introduce
the linear map from $\R^q$ to $\R^2$ given by
$$L v=\left(\begin{array}{ll}\sum_{i=1}^m v_i \\[2mm]
\sum_{i=m+1}^q v_i
\end{array}
\right).$$
In matrix form we have $L=\begin{pmatrix}I_{1,m}&0 \cr
0 & I_{1,q-m}\cr
\end{pmatrix}.$

The following lemma is straightforward.
\begin{lemma}\label{ll1}
 We have

$$L \mathbb P
=\Bigl(\begin{pmatrix} a  & 0\cr
0 & b \cr
\end{pmatrix} +
\begin{pmatrix} m p_1  & m p_2\cr
(q-m) q_2 & (q-m) q_1 \cr
 \end{pmatrix}\Bigr) L=:{\mathbb P}' L.$$

We also have
$$L \mathbb P^t
=\Bigl(\begin{pmatrix} a  & 0\cr
0 & b \cr
\end{pmatrix} +
\begin{pmatrix} m p_1  & m q_2\cr
(q-m) p_2 & (q-m) q_1 \cr
 \end{pmatrix}\Bigr) L=\hat{\mathbb P}^t L,$$
 where $\hat{\mathbb P}$ (see (\ref{TP})) is the transition
 matrix for the coarse-grained chain with two states $1',2'$. By ${\mathbb P}^t$ we denote the transpose of the matrix ${\mathbb P}$.
 \end{lemma}

The following proposition which describes all eigenvalues of ${\mathbb P}^t$
will be very important  for us.

 \begin{pro}\label{pp1} The matrix ${\mathbb P}^t$ is diagonalisable and has the eigenvalues
 $1, \lambda_2(\hat{\mathbb P}),a, b$. The dimension of the eigenspace to $a$ is $m-1$.
 The dimension of the eigenspace to $b$ is $q-m-1$.
  \end{pro}

\begin{proof}
Suppose $v$ is a right-eigenvector of ${\mathbb P}^t$ for the eigenvalue $\lambda$,
so ${\mathbb P}^t v =\lambda v$. Then using Lemma \ref{ll1} we get
$\hat{\mathbb P}^t L v =L {\mathbb P}^t v =\lambda Lv$. Two cases are possible:

Case 1:  $Lv\neq 0$ and hence $\lambda$ is an eigenvalue for $\hat{\mathbb P}$, too.
The two eigenvalues  for the two-by-two matrix can be easily
evaluated, one eigenvalue is equal to $1$, call the other one $\lambda_2(\hat{\mathbb P})$.

Case 2:  $Lv=0$.
Then we must have
\begin{equation}\begin{split}\label{transe}
{\mathbb P}^t v&\equiv {\mathbb P}^t \begin{pmatrix}u \cr
w \cr
\end{pmatrix}=
\begin{pmatrix} a E_m & 0\cr
0 & b E_{q-m}\cr
\end{pmatrix}
\begin{pmatrix}u \cr
w \cr
\end{pmatrix}\cr
&= \lambda \begin{pmatrix}u \cr
w \cr
\end{pmatrix}
\end{split}
\end{equation}
In order to have an eigenvalue $v$ at least one of the components $u$ or $w$ has
to be nonzero.

Hence the possible eigenvalues are $a, b$ (which are allowed to be equal or not).

The eigenvectors corresponding to $a$ are
of the form $ \begin{pmatrix}u \cr 0 \cr
\end{pmatrix}$ with $\sum_{i=1}^m u_i=0$. We have $m-1$ linearly independent of them.

The eigenvectors corresponding to $b$ are
of the form $ \begin{pmatrix}0 \cr w \cr
\end{pmatrix}$ with $\sum_{i=m+1}^q w_i=0$. Clearly we have $q-m-1$ of them.
\end{proof}

To check the (non-)extremality of a TISGM we will use the explicit formulas for the
eigenvalues of the matrix $\mathbb P$, which are
given by the following corollary of Proposition \ref{pp1}.

\begin{cor}\label{ca} The matrix $\mathbb P$ defined by (\ref{P}) with $m\leq [q/2]$ has the following eigenvalues
\begin{equation}\label{ev}
\begin{array}{ll}
\{1, b, \lambda_2(\hat{\mathbb P})\}, \ \ \mbox{if} \ \ m=1\\[3mm]
\{1, a, b, \lambda_2(\hat{\mathbb P})\}, \ \ \mbox{if} \ \ m\geq 2,
\end{array}
\end{equation}
where
\begin{equation}\label{abl}
a={(\theta-1)z\over Z_1}, \ \ b={(\theta-1)\sqrt[k]{z}\over Z_1}, \ \
\lambda_2(\hat{\mathbb P})={[\theta-1+(1-\sqrt[k]{z})m]z\over Z_1}.\end{equation}
\end{cor}
\begin{proof} Since $z$ is a solution to (\ref{rm}) we have $Z_1=\sqrt[k]{z} Z_2$. In case $m=1$ we have $p_1=0$ and the above mentioned condition $Lu=0$ gives $u=u_1=0$, i.e. $a$ can not be an eigenvalue. For $m\geq 2$ using $a+p_1={\theta z\over Z_1}$, $p_1={z\over Z_1}$, $b+q_1={\theta\over Z_2}$,  $q_1={1\over Z_2}$ and (\ref{lambda2}) we get (\ref{ev}).
\end{proof}
For $a,b, \lambda_2$ given by (\ref{ev}) we denote
$${\widehat\lambda}={\widehat\lambda}(k, q, m, \theta)=\max\{a, b, |\lambda_2(\hat{\mathbb P})|\}.$$
For $k=2$, using formula (\ref{s}) of $z\in \{x_1^2, x_2^2\}$, for $m=1$ we obtain $\widehat\lambda=b$ and for $m\geq 2$ we have
\begin{equation}\label{kl}
|\lambda_2(\hat{\mathbb P})|<{\widehat\lambda}=\left\{\begin{array}{ll}
b, \ \ \mbox{if} \ \ z<1\\[2mm]
a, \ \ \ \ \ \mbox{if} \ \ z>1.
\end{array}
\right.
\end{equation}

\bigskip

Let us continue now with the discussion of properties of the coarse-grained chain.

Recall that for each fixed $m$, the equation (\ref{rm}) has up to three solutions: $z_0=1, z_i=z_i(\theta,q,m), i=1,2$ (see \cite[Step 1 of the proof of Theorem 1]{KRK}).
Denote by $\mu_i=\mu_i(\theta,m)$ the TISGM of the Potts model which corresponds to the solution $z_i$, and denote by $T(\mu_i)$ its image measure, which is a Gibbs measure for the Ising model (\ref{H'}).

Define
$$g(z)=\sqrt[k]{f_m(z)}={(\theta+m-1)z+q-m\over mz+q-m-1+\theta}.$$

 By simple analysis we get the following

\begin{lemma}\label{lg} \begin{itemize}
\item[i.] For any $k\geq 2$ and $\theta>1$ the function $g(z)$, $z> 0$ has the following properties:
\begin{itemize}

\item[a)] $\{z: g^k(z)=f_m(z)=z\}=\{1, z_1, z_2\}$;

\item[b)] $a<g(z)<A, \ \ \mbox{with} \ \ a={q-m\over q+\theta-m-1}, \ \ A={\theta+m-1\over m}$;

\item[c)] ${d\over dz}g(z)={(\theta-1)(\theta+q-1)\over (mz+\theta+q-m-1)^2}>0$;

\item[d)] ${d^2\over dz^2}g(z)<0, \ \ z>0.$
\end{itemize}

\item[ii.] If $k=2$ and $m\leq q/2$ then, for the solutions $z_1$ and $z_2$ mentioned in Proposition \ref{pw}, the following statements hold
$$\begin{array}{llll}
1<z_1=z_2, \ \ \mbox{if} \ \ \theta=\theta_m\\[2mm]
1<z_1<z_2, \ \ \mbox{if} \ \ \theta_m<\theta<\theta_c, \ \ \mbox{with}\ \   \theta_c=q+1\\[2mm]
1=z_1<z_2, \ \ \mbox{if} \ \ \theta=q+1\\[2mm]
z_1<1<z_2, \ \ \mbox{if} \ \ q+1<\theta.
\end{array}$$
\end{itemize}
\end{lemma}

In Figure \ref{fn1} the functions $z_i=x^2_i(m,\theta)$, $i=1,2$ are shown for $k=2$, $q=8$, $m=1,2,3$.

\begin{figure}
\includegraphics[width=10cm]{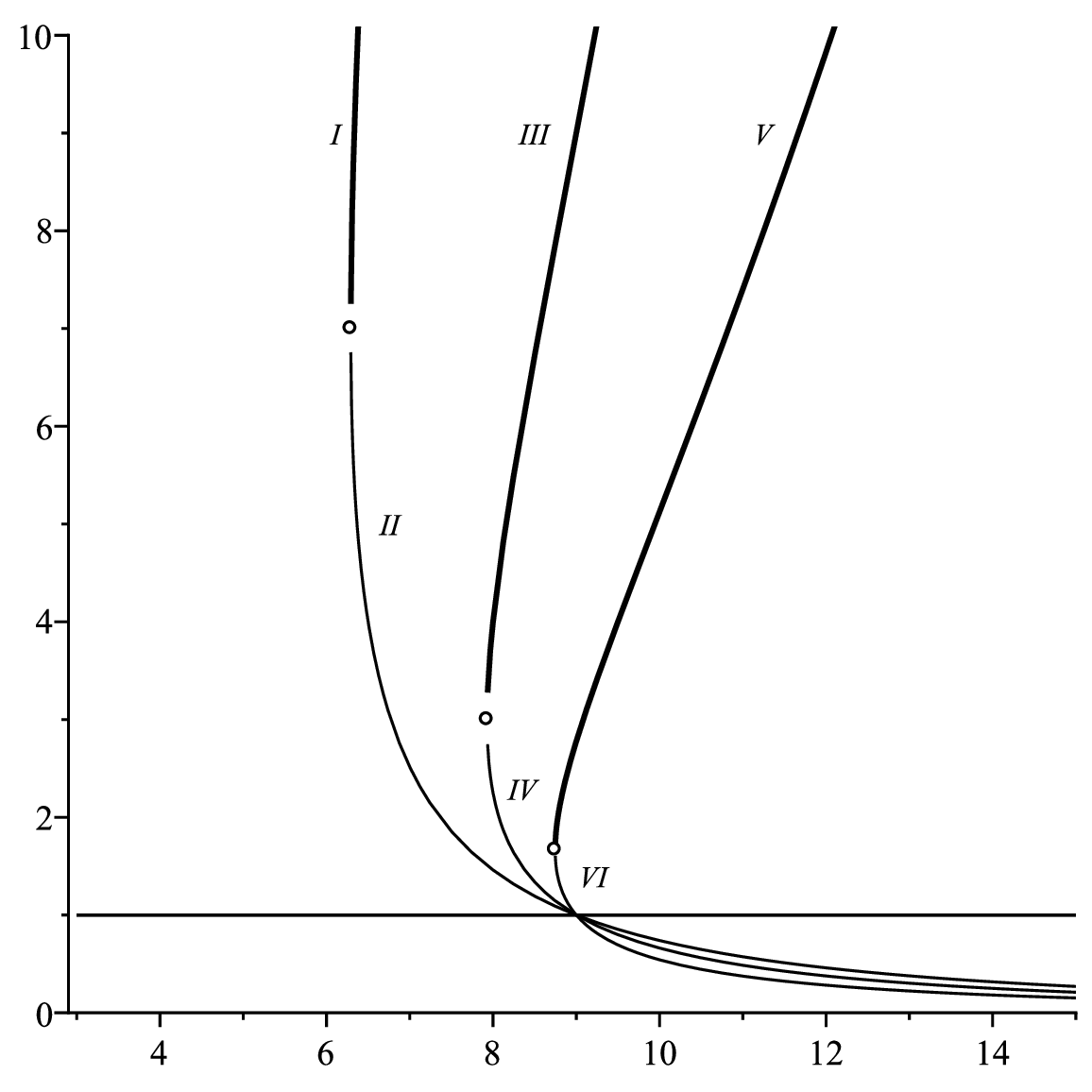}
\caption{ The graphs of the functions $z_i=z_i(m,\theta)$, $i=1,2$, for $k=2$, $q=8$ and $m=1,2,3$. The horizontal axis presents $\theta$ and the vertical axis presents $z_i$. Circle dots having coordinates $(\theta_m, {q-m\over m})$, $m=1,2,3$ separate graph of $z_1$ from graph of $z_2$. The graph of $z_2(1,\theta)$ is the curve $I$. The graph of $z_1(1,\theta)$ is the curve $II$. $z_2(2,\theta)$ is $III$. $z_1(2,\theta)$ is $IV$, $z_2(3,\theta)$ is $V$. $z_1(3,\theta)$ is $VI$. The intersection point of $z_1$ functions has coordinate $(\theta_c,1)=(9,1)$. For $m=4$ the graphs of the functions $z_i(4,\theta)$ are similar to the above given graphs, but the circle point separating them coincides with the point $(9,1)$ of the intersections, since $\theta_{\lfloor{8\over 2}\rfloor=4}=\theta_c=9$.}\label{fn1}
\end{figure}

  We have the following

\begin{pro}\label{bh}
\begin{itemize}
\item[1.] Each SGM of the model (\ref{H'}) corresponds\footnote{This correspondence is as in the definition of a SGM,  see  below of the formula (\ref{dlr}).} to a solution $z_x\in R$, $x\in V$ of the following functional equation
\begin{equation}\label{Tz}
z_x=\prod_{y\in S(x)} {(\theta+m-1)z_y+q-m\over mz_y+q-m-1+\theta}.
\end{equation}

\item[2.] If $z_x$ is a solution of (\ref{Tz}) then
\begin{equation}\label{To}
\min\{1, z_1, z_2\}\leq z_x\leq \max\{1, z_1, z_2\},
\end{equation}
where $z_1$, $z_2$ are solutions to (\ref{rm}).
\end{itemize}
\end{pro}
\begin{proof}
1. This was proved already below \eqref{transe466} using relations between boundary laws.
One can also give a direct proof, as in \cite[Theorem 2.1.]{Ro}.

2. For $z>0$ by Lemma \ref{lg} we have $a<g(z)<A$.
Using this inequality from (\ref{Tz}) we get
$$a^k< z_x< A^k.$$
 Now we consider the function $g(z)$ on $[a^k, A^k]$ and on this segment, using Lemma \ref{lg}, we get the bounds $g(a^k)<g(z)<g(A^k)$. Then
$$(g(a^k))^k<z_x<(g(A^k))^k.$$
 Iterating this argument we obtain
$$f_m^n(a^k)< z_x< f_m^n(A^k),$$
where $f_m(z)$ is defined in (\ref{rm}), and $f_m^n$ is its $n$th iteration.
It is easy to see that $\max\{1,z_1,z_2\}\leq f_m^n(A^k)$ and
the sequence  $a_n=f_m^n(A^k)$ is decreasing. Indeed, by Lemma \ref{lg} the function
$g$ is increasing, hence $f_m(z)=(g(z))^k$ is an  increasing function too and $f_m$ is bounded
by $A^k$, i.e., $f_m(z)<A^k$. Writing the last inequality for $z=A^k$ we get $a_1<A^k$.
Since $f_m$ is an increasing function  we obtain from the last inequality $a_2=f_m(a_1)< f_m(A^k)=a_1$, i.e.
$a_2<a_1$. From this iterating $f_m$ again we get $a_3<a_2$ and so on.
 Thus the sequence $a_n$ has a limit $\alpha$, with $\alpha\geq \max\{1,z_1,z_2\}$.
This limit point must be a fixed point for $f_m$. But since the function $f_m$ has no fixed point in $(\max\{1,z_1,z_2\},+\infty)$ we get that $\alpha=\max\{1,z_1,z_2\}$.
\end{proof}

Let $\mathcal G(H')$ the set of all SGMs of the model (\ref{H'}).

Denote
\begin{equation}\label{te1}
\theta_m(k)=1+{mx_*^k+q-m\over x_*^{k-1}+x_*^{k-2}+\dots+x_*},
\end{equation}
where $x_*$ is the unique solution (see Step 3 in proof of Theorem 1 \cite{KRK}) of the following equation
\begin{equation}\label{se}
m\sum_{i=1}^{k-1}ix^{2k-i-1}-(q-m)\sum_{i=1}^{k-1}ix^{i-1}=0.
\end{equation}

\begin{thm}\label{tii} For the Ising model (\ref{H'}) on the Cayley tree of order $k\geq 2$ the following statements are true
 \begin{itemize}
    \item[(1)] If $\theta < \theta_m(k)$ then there is unique TISGM $T(\mu_0)$.

    \item[(2)] If $\theta=\theta_m(k)$ or $\theta=\theta_c={q+k-1\over k-1}$ then there are 2 TISGMs $T(\mu_0), T(\mu_1)$.

    \item[(3)] If $\theta_m(k)<\theta\ne \theta_c$ then there are 3 TISGMs $T(\mu_0), T(\mu_1), T(\mu_2)$.
    Moreover, at least two of these measures are extreme in $\mathcal G(H')$.
    \end{itemize}
 \end{thm}
 \begin{proof}\footnote{This theorem is a modification of well-known Theorem 12.31 of \cite{Ge}.} The measures $T(\mu_i)$, $i=0,1,2$ correspond to the solutions $z_0=1$, $z_1$ and $z_2$ of (\ref{rm}).

 (1) In case $\theta < \theta_m(k)$ the equation (\ref{rm}) has the unique solution $z=z_0=1$. Consequently, by part 2 of Proposition \ref{bh}
 we obtain that $z_x\equiv 1$ is a unique solution of (\ref{Tz}). Thus $\mathcal G(H')=\{T(\mu_0)\}$.

 (2) If $\theta=\theta_m(k)$ then $z_1=z_2$, i.e. there are only two TISGMs $T(\mu_0), T(\mu_1)$; if $\theta=\theta_c={q+k-1\over k-1}$ then $z_0=z_1$
 and there are again  two TISGMs.

 (3) If $\theta_m(k)<\theta\ne \theta_c$ then all solutions $z_0=1$, $z_1$ and $z_2$ of (\ref{rm}) are distinct. The extremality of two of these measures can be deduced using the minimality and maximality of the corresponding values
of solutions. Assume that $z_2$ is the maximal solution and $T(\mu_2)$ is non-extreme,
with a decomposition
$$T(\mu_2) =\int T(\mu)(\omega)\nu(d\omega).$$

Then for any vertex $x\in V$ we have
\begin{equation}\label{em}
z_2 =\int z_x(\omega)\nu(d\omega).
\end{equation}

By (\ref{To}) and Lemma \ref{lg} $z_2$ is an extreme point in the set of all solutions $z_x$,
and so (\ref{em}) holds if $z_x(\omega) =z_2$ for almost all $\omega$. Hence, $T(\mu_2)$ is extreme.
  \end{proof}

\section{Conditions for non-extremality}

Let us continue with the investigation of the coarse-grained chains $T(\mu)$ with
a focus on criteria for non-extremality which are based on properties of the corresponding
$2\times 2$-transition matrix as it depends on coupling strength of the  Potts model,
the block-size $m$ and the choice of a  branch of the boundary law.

In the case of three solutions, we denote the middle solution by $z_{\rm mid}$, which is unique element of the following set
$$\{1,z_1,z_2\}\setminus \{\min\{1,z_1,z_2\}, \max\{1,z_1,z_2\}\}.$$
By Lemma \ref{lg}, for $2m<q$ we have $z_{\rm mid}=\max\{1,z_1\}$.

Let $T(\mu_{\rm mid})$ be the TISGM which corresponds to $z_{\rm mid}$.
Recall that $\mu_i$ is the TISGM of the Potts model which corresponds to the boundary law $z_i(m)$, for the branches
$i=1,2$.

Define the following numbers:
$$\theta_0=1+q+2\sqrt{2m(q-m)}, \ \ \hat\theta_0=1+(\sqrt{2}+1)q,$$
\begin{equation}\label{t*}
\widehat\theta=(\sqrt{2}-1)q+2m+1, \ \ \theta^*=1+(\sqrt{2}+1)q-2m.
\end{equation}

\begin{thm}\label{tne} Let $k=2$, $2m<q$. Then the following statements hold.
\begin{itemize}
 \item[(i)] If $\theta\geq q+1$ then $T(\mu_{\rm mid})=T(\mu_0)$ and
$$T(\mu_0) \hbox{ is } \left\{\begin{array}{ll}
{\rm extreme}, \ \ \mbox{if} \ \ \theta\leq \theta_0\\[3mm]
{\rm non-extreme}, \ \ \mbox{if} \ \ \theta>\hat\theta_0;
\end{array}\right.
$$
\item[(ii)] Assume one of the following conditions is satisfied:

\begin{itemize}
\item[a)] $2\leq m\leq q/7$ and $\theta\in [\theta_m, \widehat\theta)$;

\item[b)] $\theta\in (\theta^*,+\infty)$.
\end{itemize}

Then  $\mu_1(\theta,m)$ is non-extreme.

\item[(iii)] Assume one of the following conditions is satisfied:

\begin{itemize}
\item[c)] $2\leq m\leq q/7$ and $\theta\geq \theta_m$;

\item[d)]  $q<7m$, $m\geq 2$ and $\theta\in (\widehat\theta, +\infty)$.

\end{itemize}

Then  $\mu_2(\theta,m)$ is non-extreme.(See Fig.\ref{fn2}-\ref{fn4})

\item[(iv)] If $\theta_m\leq \theta<q+1$ then $T(\mu_{\rm mid})=T(\mu_1)$. Moreover, if $q>85$ then for each $m\leq [{q\over 85}]$ there are two critical values $ \bar\theta, \, \bar{\bar\theta}\in (\theta_m, q+1)$, with $\bar\theta<\bar{\bar\theta}$ such that $T(\mu_1)$ is non-extreme if $\theta\in (\bar\theta, \bar{\bar\theta})$. Moreover, $\bar\theta, \bar{\bar\theta}$ are solutions to
    \begin{equation}\label{ke}
(\theta-1)^3-(\sqrt{2}-1)q(\theta-1)^2-2(2\sqrt{2}-1)m(q-m)(\theta-1)+2qm(q-m)=0.
\end{equation}
\end{itemize}
\end{thm}
\begin{figure}
\includegraphics[width=8cm]{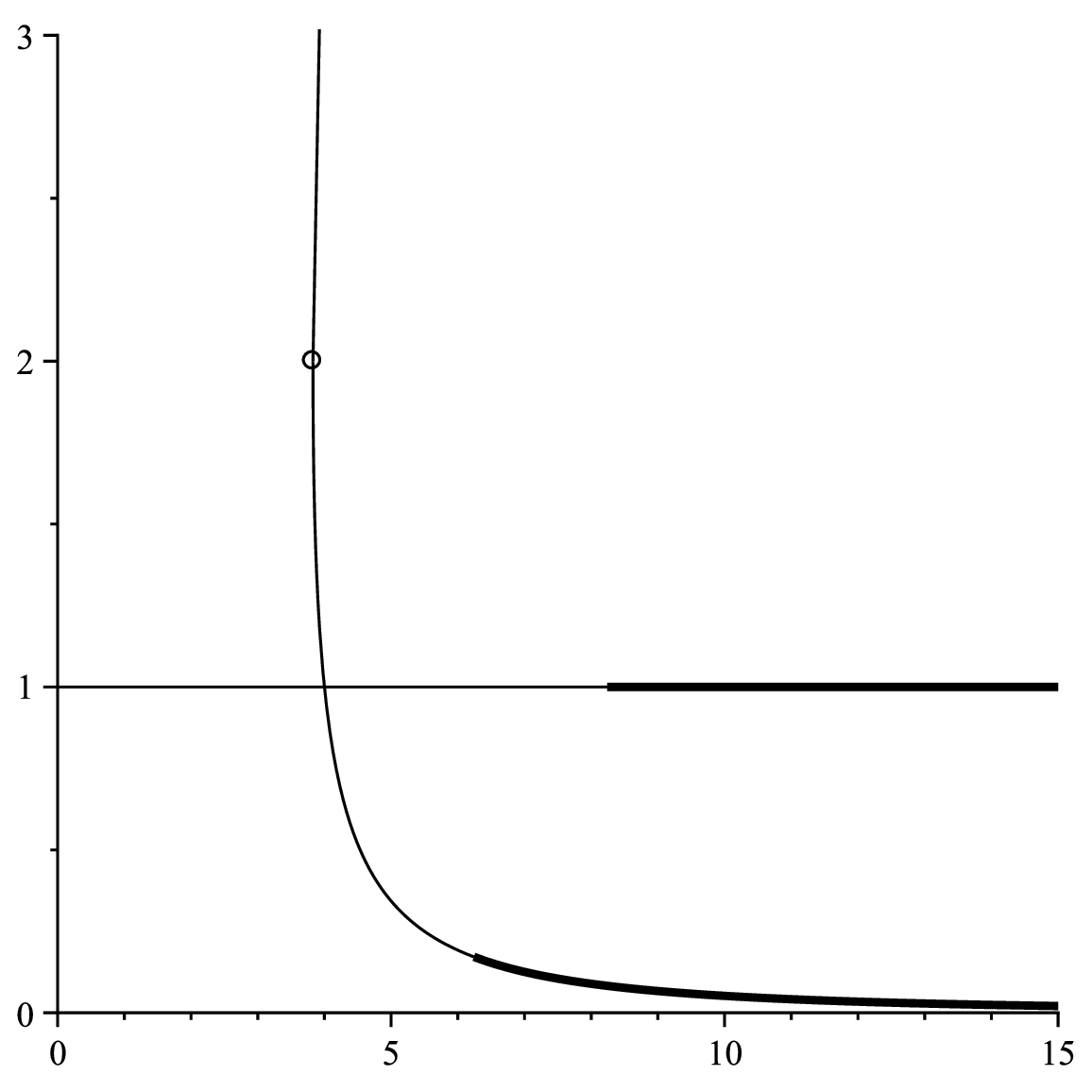}
\caption{ The graphs of the functions $z_i=z_i(m,\theta)$, $i=1,2$, for $q=3$ and $m=1$. The circle dot having coordinate $(\theta_1, 2)$ separates graph of $z_1$ from graph of $z_2$. The bold curves correspond to regions of solutions where the corresponding TISGM is non-extreme. This figure corresponds to Part (ii), b) of Theorem \ref{tne}.}\label{fn2}
\includegraphics[width=8cm]{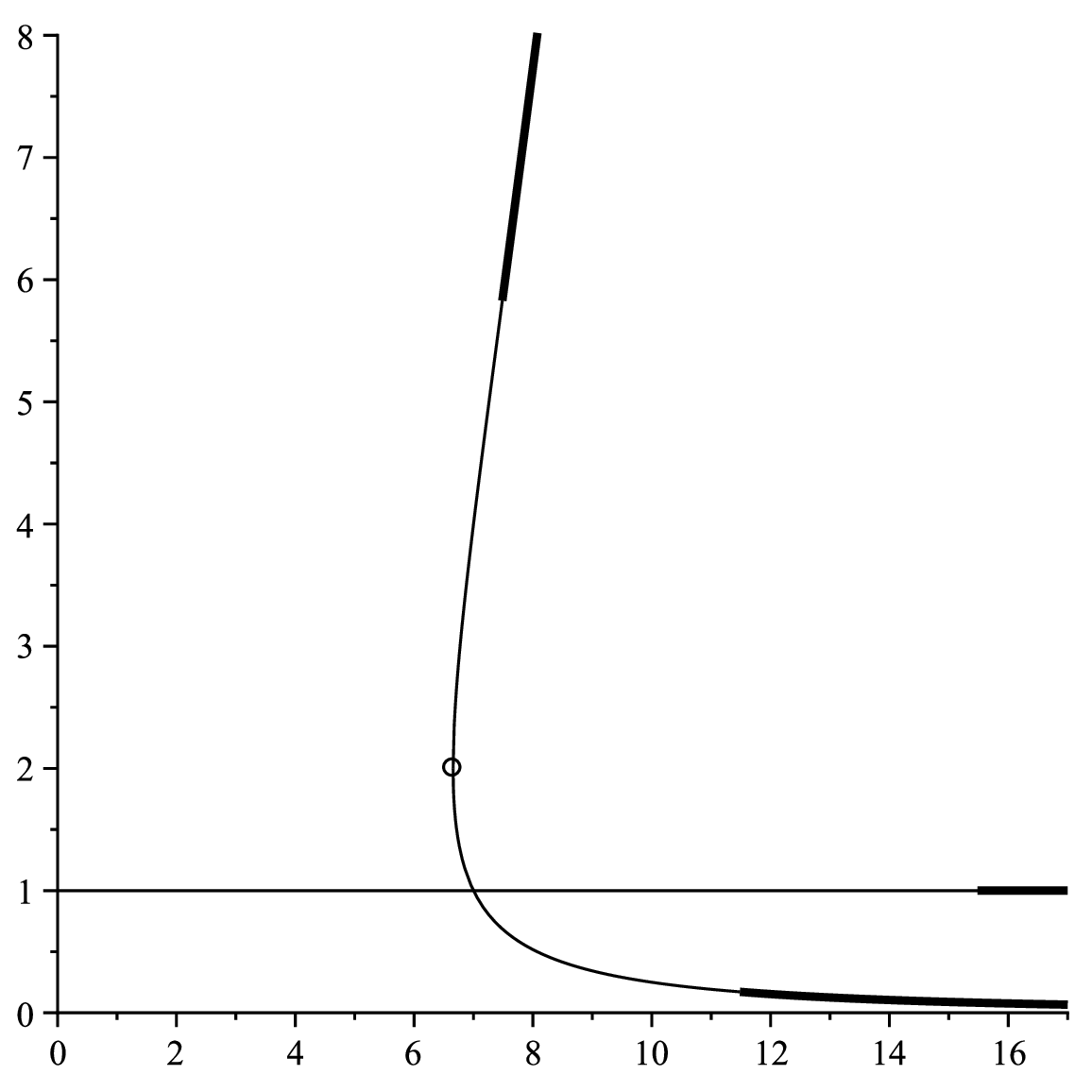}
\caption{The graphs of the functions $z_i=z_i(m,\theta)$, $i=1,2$, for $q=6$ and $m=2$. The circle dot having coordinate $(\theta_2, 2)$ separates graph of $z_1$ from graph of $z_2$. The bold lines correspond to regions of solutions where corresponding TISGM is non-extreme. The upper bold curve corresponds to part (iii), d) and the lower bold curve corresponds to part (ii), b) of Theorem \ref{tne}.}\label{fn3}
\end{figure}
 \begin{figure}
\includegraphics[width=7cm]{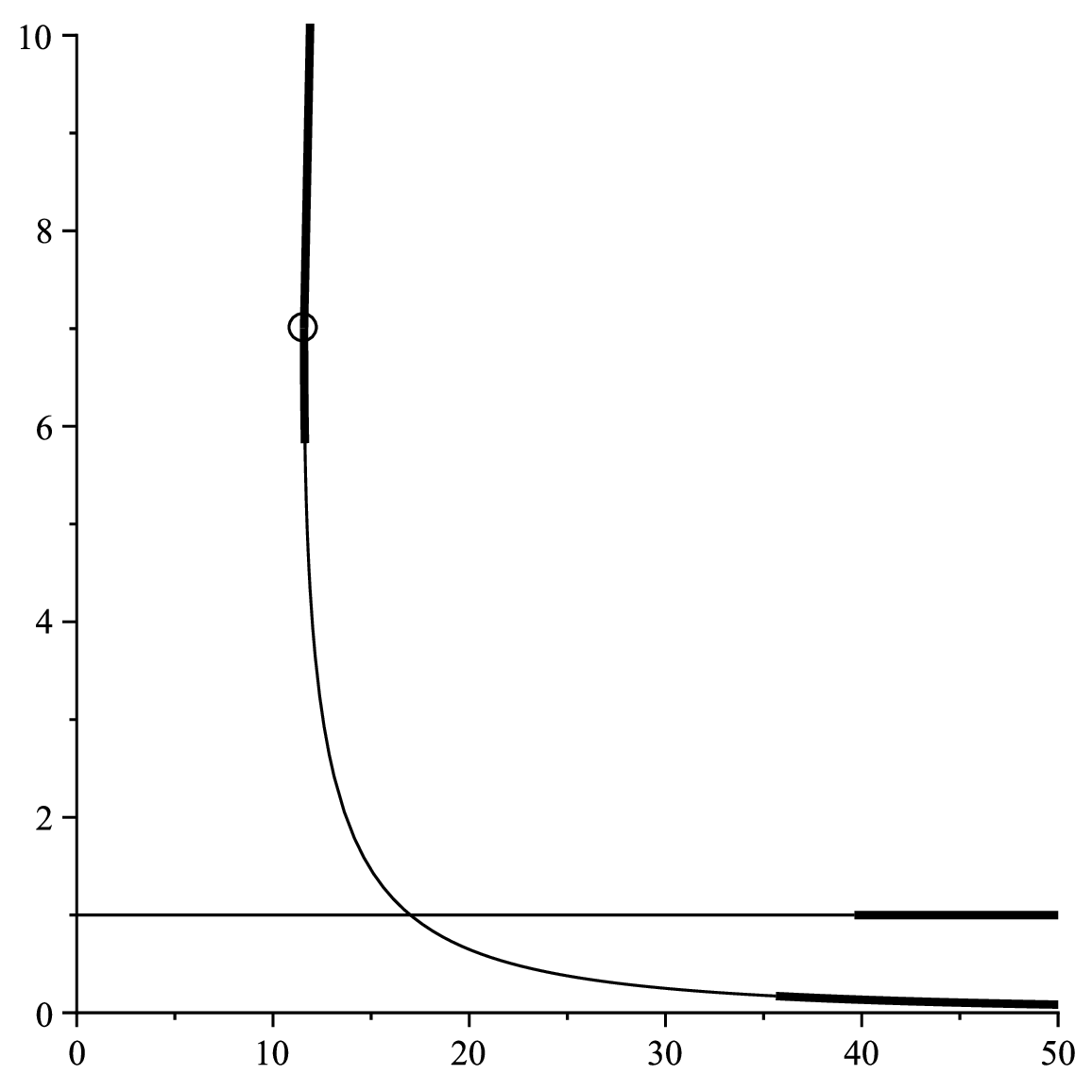}
\caption{The graphs of the functions $z_i=z_i(m,\theta)$, $i=1,2$, for $q=16$ and $m=2$. The circle dot having coordinate $(\theta_2, 7)$ separates graph of $z_1$ from graph of $z_2$. The bold lines correspond to regions of solutions where corresponding TISGM is non-extreme. The bold $z_2$ corresponds to Part (iii), c); Bold parts of $z_1$ correspond to Part (ii), a) (upper bold curve) and Part (ii), b) (lower bold curve) of Theorem \ref{tne}.}\label{fn4}
\end{figure}

\begin{proof} (i) The equality $T(\mu_{\rm mid})=T(\mu_0)$ follows from the second part of Lemma \ref{lg}. To check the extremality we apply arguments used for the reconstruction on trees \cite{FK}, \cite{Ke}, \cite{Mar}, \cite{Mos2}, \cite{Mos}. Consider Markov chains
with states $\{1',2'\}$ and transition probabilities $p_{ij}$ defined by (\ref{TP}).
It is known that a sufficient condition for non-extremality (which is equivalent to solvability of the associated reconstruction) of a Gibbs measure $T(\mu)$ corresponding
to the matrix $T({\mathbb P})$ is that
$k\lambda^2_2>1$, where $\lambda_2$ is the second largest (in absolute value) eigenvalue of $T({\mathbb P})$ \cite{Ke}.
On the other hand, Martin in [3] gives the following condition for extremality (non-reconstructibility)
$$k\left(\sqrt{T({\mathbb P})_{11}T({\mathbb P})_{22}}-\sqrt{T({\mathbb P})_{12}T({\mathbb P})_{21}}\right)^2\leq 1.$$

 In case $z=1$ the matrix $T({\mathbb P})$ is
\begin{equation}\label{ce}
T({\mathbb P})={1\over \theta+q-1}\left(\begin{array}{cc}
\theta+m-1 & q-m\\[2mm]
m & \theta+q-m-1
\end{array}\right).
\end{equation}

Simple calculations show that above-mentioned conditions of extremality for $T(\mu_0)$
(i.e. for the matrix (\ref{ce})) are equivalent to the conditions on $\theta$ as mentioned in theorem.

(ii) {\bf Case $m\geq 2$.}

{\it Subcase}: $\theta_m\leq\theta<q+1$. In this case by Lemma \ref{lg} we have $z_1>1$.
By (\ref{kl}) (which requires $m\geq 2$) and Corollary \ref{ca} we have the second largest (in absolute value) eigenvalue of $\mathbb P$
 is $a$ so to prove non-extremality we should check the Kesten-Stigum condition $2a^2>1$.
Since $a>0$ the last inequality is equivalent to $\sqrt{2}a>1$.
Denote
$$\gamma_1(\theta)=\sqrt{2}a-1=\sqrt{2}{(\theta-1)z_1\over (\theta+m-1)z_1+q-m}-1,$$
where
 $$z_1=\left({\theta-1-\sqrt{(\theta-1)^2-4m(q-m)}\over 2m}\right)^2.$$
 We have
 $$\gamma_1'(\theta)={\sqrt{2}(\theta-1)z_1\over ((\theta+m-1)z_1+q-m)^2}\left(\sqrt{z_1}-{2(q-m)\over \sqrt{(\theta-1)^2-4m(q-m)}}\right)$$
 $$=-\,{\sqrt{2z_1}(\theta-1)^2z_1\over ((\theta+m-1)z_1+q-m)^2\sqrt{(\theta-1)^2-4m(q-m)}}.$$
 Thus for any $\theta\geq \theta_m$ we
have $\gamma_1'(\theta)<0$. Hence $\gamma_1$ is a decreasing function of $\theta$. So to have $\sqrt{2}a>1$ it is necessary
that
$$\gamma_1(\theta_m)=\sqrt{2}{\sqrt{q-m}\over \sqrt{q-m}+\sqrt{m}}-1>0,$$ which is satisfied for $q\geq 7m$. Moreover, we have $\gamma_1(q+1)=-({\sqrt{2}-1\over\sqrt{2}})<0$. Consequently, there exists a unique $\theta=\widehat\theta$ such that $\gamma_1(\widehat\theta)=0$. To find $\widehat\theta$ we solve  $\gamma_1(\widehat\theta)=0$  and by long (but simple) computations, using the explicit formula for $\gamma_1(\theta)$ and the expression for $z_1$ we get
$$\widehat\theta=(\sqrt{2}-1)q+2m+1.$$

Thus we proved that $2a^2>1$ iff $\theta<\widehat\theta$.

{\it Subcase:} $\theta\geq q+1$. In this case by Lemma \ref{lg} we have $z_1<1$ so by (\ref{kl}) and Corollary \ref{ca} we must check $2b^2>1$.
Define
$$\xi_1(\theta)=\sqrt{2}b-1={(\theta-1)\sqrt{2z_1}\over (\theta+m-1)z_1+q-m}-1,$$
 We have
 $$\xi_1'(\theta)={\sqrt{2z_1}(\theta-1)^2z_1\over ((\theta+m-1)z_1+q-m)^2\sqrt{(\theta-1)^2-4m(q-m)}}>0.$$
  Hence $\xi_1$ is an increasing function of $\theta$.  Moreover, we have $\xi_1(q+1)=\gamma_1(q+1)=-({\sqrt{2}-1\over\sqrt{2}})<0$.
Since $\xi_1(\theta)$ is increasing it is bounded from above by its
limit for  $\theta\to +\infty$. So we compute
$$\lim_{\theta\to\infty}\xi_1(\theta)=-1+\sqrt{2}\lim_{\theta\to\infty}{(\theta-1)\sqrt{z_1}\over (\theta+m-1)z_1+q-m}$$
\begin{equation}\label{lim}
=-1+\sqrt{2}\lim_{\theta\to\infty}{1\over \left({\theta+m-1\over \theta-1}\right)\sqrt{z_1}+{q-m\over (\theta-1)\sqrt{z_1}}}.\end{equation}
Using formula (\ref{s}) we get
$$\lim_{\theta\to\infty}\sqrt{z_1}=\lim_{\theta\to\infty}{2(q-m)\over \theta-1+\sqrt{(\theta-1)^2-4m(q-m)}}=0.$$
$$\lim_{\theta\to\infty}(\theta-1)\sqrt{z_1}=\lim_{\theta\to\infty}{2(q-m)\over 1+\sqrt{1-{4m(q-m)\over (\theta-1)^2}}}=q-m.$$
Using these formulas we get from (\ref{lim})
$$\lim_{\theta\to\infty}\xi_1(\theta)=\sqrt{2}-1>0.$$
  Consequently, there exists a unique $\theta=\theta^*$ such that $\xi_1(\theta^*)=0$. To find $\theta^*$ we solve  $\xi_1(\theta^*)=0$ and by simple computations using the
  explicit formula of $\xi_1(\theta)$ we get
$$\theta^*=1+(\sqrt{2}+1)q-2m.$$

Thus we proved that $2b^2>1$ if $\theta>\theta^*$.

{\bf Case $m=1$.} In this case, as mentioned before (\ref{kl}) the second largest eigenvalue is $b$,
therefore we have to check $\xi_1(\theta)=\sqrt{2}b-1>0$. We note that $\xi_1'(\theta)>0$ for each $\theta\geq \theta_m$, and
$$\xi_1(\theta_m)={\sqrt{2m}\over \sqrt{q-m}+\sqrt{m}}-1<0.$$
Hence again we have $\xi_1(\theta)>0$ iff $\theta\in (\theta^*,+\infty)$.

(iii) {\bf Case $m\geq 2$.} Since $z_2>1$ for any $\theta\geq \theta_m$, using (\ref{kl}) and Corollary \ref{ca} we shall check only $2a^2>1$. Denote
$$\gamma_2(\theta)=\sqrt{2}a-1=\sqrt{2}{(\theta-1)z_2\over (\theta+m-1)z_2+q-m}-1,$$
where
 $$z_2=\left({\theta-1+\sqrt{(\theta-1)^2-4m(q-m)}\over 2m}\right)^2.$$
 We have
 $$\gamma_2'(\theta)={\sqrt{2}(\theta-1)z_2\over ((\theta+m-1)z_2+q-m)^2}\left(\sqrt{z_2}+{2(q-m)\over \sqrt{(\theta-1)^2-4m(q-m)}}\right)>0.$$
 Thus for any $\theta\geq \theta_m$ we
have $\gamma_2'(\theta)>0$. Hence $\gamma_2$ is an increasing function of $\theta$. For $q\geq 7m$ we have
$$\gamma_2(\theta_m)=\gamma_1(\theta_m)>0.$$ Hence $\gamma_2(\theta)>0$ for all $\theta>\theta_m$. Consequently, $\mu_2$ is non-extreme.

In case $q\leq 6$ or $q\geq 7$ with $[q/7]<m$ we have $\gamma_1(\theta_m)=\gamma_2(\theta_m)<0$.  Since $\gamma_2$ is an increasing function $\gamma_2(\theta)=0$ has a unique solution, which is equal to $\widehat \theta$. Thus $\gamma_2(\theta)>0$ for all $\theta>\widehat\theta$.

{\bf Case $m=1$.} Define
$$\xi_2(\theta)=\sqrt{2}b-1={(\theta-1)\sqrt{2z_2}\over \theta z_2+q-1}-1.$$
 We have
 $$\xi_2'(\theta)=-{\sqrt{2z_2}(\theta-1)^2z_2\over (\theta z_2+q-1)^2\sqrt{(\theta-1)^2-4(q-1)}}<0.$$
  Hence $\xi_2$ is a decreasing function of $\theta$.  Moreover, we have
  $$\xi_2(\theta_m)=\xi_1(\theta_m)=-(1-{\sqrt{2}\over\sqrt{q-1}+1})<0.$$
  Thus $\xi_2(\theta)<0$ for all $\theta\geq \theta_m$.

(iv) In this case the matrix $T({\mathbb P})$ is
\begin{equation}\label{ce1}
T({\mathbb P})=\left(\begin{array}{cc}
{(\theta+m-1)z\over Z_1}  & {q-m\over Z_1}\\[3mm]
{mz\over Z_2}  & {\theta+q-m-1\over Z_2}
\end{array}\right),
\end{equation}
which has eigenvalues $1$ and $\lambda_2(\hat{\mathbb P})$ given in (\ref{abl}).
It is easy to check that
$$\theta-1-(\sqrt{z_i}-1)m\geq 0, \ \ i=1,2,$$ i.e. $\lambda_2(\hat{\mathbb P})>0.$

We define the following function
$$\eta_1(\theta)=\sqrt{2}{(\theta-1-(\sqrt{z_1}-1)m)z_1\over (\theta+m-1)z_1+q-m}-1.$$
We have
$$\eta_1(\theta_m)=\eta_1(q+1)=-{\sqrt{2}-1\over \sqrt{2}}<0.$$

Using the explicit formulas for $z_1$ (given in Proposition \ref{pw}) we note that  $\eta_1(\theta)>0$ iff
\begin{equation}\label{ke1}
\psi(\theta)=(\theta-1)^3-(\sqrt{2}-1)q(\theta-1)^2-2(2\sqrt{2}-1)m(q-m)(\theta-1)+2qm(q-m)<0.
\end{equation}
It is well known (see \cite[p.28]{Pra})\footnote{This is known as the Descartes rule: The number of positive roots of the polynomial $p(x)=a_0x^n+a_1x^{n-1}+\dots+a_n$ does not exceed the number of sign changes in the sequence $a_0, a_1,\dots,a_n$.} that the number of positive
roots of a polynomial does not exceed the number of sign
changes of its coefficients. Using this fact one can see that the
equation $\psi(\theta)=0$ has up to two positive roots denoted by $\bar\theta, \bar{\bar\theta}$.
Moreover $\psi'(\theta)=0$ has the unique positive solution
$$\theta_0=1+{q\over 3}(\sqrt{2}-1)+{1\over 3}\sqrt{(3-2\sqrt{2})q^2+6(2\sqrt{2}-1)m(q-m)}.$$
Since $\theta>1$, $\psi(1)>0$ and $\psi(\infty)>0$ we have that at $\theta_0$
the function $\psi(\theta)$ has minimum. Therefore to have
non-empty set of solutions for $\psi(\theta)<0$ it is necessary $\psi(\theta_0)<0$.
Maple solved $\psi(\theta_0)=0$ with respect to $q$ and gave explicitly four solutions:
$$q=m,\ \ q=2m, \ \ q=m_1, \ \ q=m_2,$$
where $m_1$ is negative and $m_2$ is a linear function of $m$, i.e.,
$$m_2=\frac{14\left(14308\sqrt{2}+19845+288\sqrt{2}\sqrt{1055+746\sqrt{2}}+452\sqrt{1055+746\sqrt{2}}\right)}
{4015+4768\sqrt{2}-1127\sqrt{1055+746\sqrt{2}}+833\sqrt{2}\sqrt{1055+746\sqrt{2}}}m.$$
We note that $m_2\approx 85m$. By our assumption $1\leq m\leq q/2$, the condition
$\psi(\theta_0)<0$ is equivalent to $q>m_2$. The last inequality holds for some $m$ (with $1\leq m\leq q/2$)
 if $q>85$.
 This completes the proof.
\end{proof}
\begin{rk}
In the case $q=7$ the condition $q\geq 7m$ leaves only the value $m=1$. Then $\theta_1=1+2\sqrt{6}\approx 5.898979$ and
$\widehat\theta=3+7(\sqrt{2}-1)\approx 5.899494$. Hence $\widehat\theta-\theta_1\approx 0.00051$.
\end{rk}

\section{Conditions for extremality}

We turn our attention to sufficient conditions for extremality (or non-reconstructability in information-theoretic language)
of the full chains of the Potts model,
depending on coupling strength parameterized by $\theta$, the block size $m$ and the branch of the boundary
law $z$.

Recall that by $\mu_i(\theta,m)$ we denote the TISGM which corresponds to
the values of the boundary law $z_i(\theta,m)$, $i=1,2$, which is a solution to (\ref{rm}).\\

{\bf 5.1. Main results of this section.} Let us first give all main results of this section.

Recall $\theta_1=1+2\sqrt{q-1}$ (see (\ref{tm})) and $\theta^*=1+(\sqrt{2}+1)q-2m$ (see (\ref{t*})).

\begin{thm}\label{t1} If $k=2$, $m=1$ then the following is true.
\begin{itemize}
\item[(a)] \begin{itemize}
\item If $q=3,4,\dots,16$ then there exists $\theta^{**}$ such that $\theta_c=q+1<\theta^{**}<\theta^*$ and the
measure $\mu_1(\theta,1)$ is extreme for any $\theta\in [\theta_1, \theta^{**})$.
Moreover $\theta^{**}$ is the unique positive solution of the following equation
$$\theta^3-(q-3)\theta^2-(2q-7)\theta-(q+5)=0.$$
\item If $q\geq 17$ then there are $\bar\theta_1, \bar\theta_2\in (\theta_1,\theta_c)$ such that $\bar\theta_1<\bar\theta_2$ and
the measure $\mu_1(\theta,1)$ is extreme for any $\theta\in [\theta_1, \bar\theta_1)\cup (\bar\theta_2,\theta^{**})$.
Moreover $\bar\theta_1, \bar\theta_2$ are positive solutions of the following equation
$$\theta^3-(q-1)\theta^2-(2q-3)\theta+(4q^2-13q+11)=0.$$
\end{itemize}
\item[(b)] The measure $\mu_2(\theta,1)$ is extreme for any $\theta\geq  \theta_1$, $q\geq 2$. (see Fig.\ref{fn5})
\end{itemize}
\end{thm}

\begin{rk} We are considering the solution of the form $(\underbrace{z,z,\dots,z}_m,1,\dots,1)$, but of course
 the results are true for any permutation of the coordinates of this solution. We note that $\mu_2(\theta,1)$ corresponds to $m=1$, by the permutations we get $q$ distinct measures similar to $\mu_2(\theta,1)$. These measures coincide with the well-known TISGMs considered in \cite{Ga8}, because in \cite{Ga8} the author only used the maximal solutions $z_2$. Moreover in \cite{Ga8} it was proved that these measures are extreme as soon as they exist. Our proof (that is the proof of the part (b) of Theorem \ref{t1}) is an alternative to the proof of \cite{Ga8}. Since the solutions $z_1=z_1(\theta,q,m)$ for $m\geq 1$ and $z_2=z_2(\theta,q,m)$ for $m\geq 2$ were not known before, the extreme measures
 $\mu_1(\theta,m)$, $\mu_2(\theta,m)$  corresponding to them which we
 constructed in our paper, are new.
\end{rk}

The following theorem corresponds to the case: $m\geq 2$. From condition $2\leq m\leq [q/2]$ it follows that $q\geq 4$.
\begin{thm}\label{t3} \begin{itemize} Let $k=2$.
\item[(i)] If $m=2$ then the following is true.
\begin{itemize}
\item For each $q=4,5,6,7,8$ there exists $\breve\theta>\theta_c=q+1$ such that the measure $\mu_1(\theta,2)$ is extreme for any $\theta\in [\theta_2, \breve\theta)$.
\item For each $q\geq 9$ there exists $\theta^\dag\in (\theta_2,q+1)$ such that the measure $\mu_1(\theta,2)$ is extreme for any $\theta\in [\theta^\dag, \breve\theta)$, where $\theta^\dag=\theta^\dag(q)$ is the unique solution of
    $$ \theta^3-(q+3)\theta^2+(6q-17)\theta-(9q-19)=0$$ and $\breve\theta=\breve\theta(q)$ is the unique solution of
    $$\theta^3-(q+3)\theta^2-(2q-15)\theta-(q+13)=0.$$
    \end{itemize}
\item[(ii)] If $m=2$ then for each $q=4,5,6,7,8$ there exists $\grave{\theta}=\grave{\theta}(q)$ such that $\theta_2<\grave{\theta}\leq q+1$ and $\mu_2(\theta,2)$ is extreme for $\theta\in [\theta_2, \grave\theta)$ (see Fig.\ref{fn7}).
\item[(iii)] If $q<{m+1\over 2m}\left[3m+1+\sqrt{m^2+6m+1}\right]$ and $m\geq 2$ then the measure $\mu_1(\theta_m,m)=\mu_2(\theta_m,m)$ is extreme.
\end{itemize}
\end{thm}
The following theorems are obtained by a continuity argument.
\begin{thm}\label{th2}\begin{itemize} Let $k=2$.
 \item[a)] For each $m\leq [q/2]$ there exists a neighborhood $U_m(\theta_c)$ of $\theta_c$ such that the
measure $\mu_1(\theta,m)$ is extreme if $\theta\in U_m(\theta_c)$ (see Fig.\ref{fn9}).

\item[b)] If $m=1$ or condition (iii) of Theorem \ref{t3} is satisfied then there exists a neighborhood $V_m(\theta_m)$ of $\theta_m$ such that
measures $\mu_i(\theta,m),$ $i=1,2$ are extreme if $\theta\in V_m(\theta_m)$.
\end{itemize}
\end{thm}
\begin{thm}\label{th8} For the ferromagnetic
$q$-state Potts model (with $q\geq 3$) on the Cayley tree
of order two, there exists a punctured neighborhood $U(\theta_c)$ of $\theta_c$
such that there are at least $2^{q-1}+q$ extreme TISGMs for each $\theta\in U(\theta_c)$.
\end{thm}

Here the model is called ferromagnetic if $\theta>1$, and as usual
a punctured neighborhood of $\theta_c$ is an open neighborhood from which the value
$\theta_c$ was removed.

The proofs of these theorems follow after the following subsection.\\

\begin{figure}
\includegraphics[width=8cm]{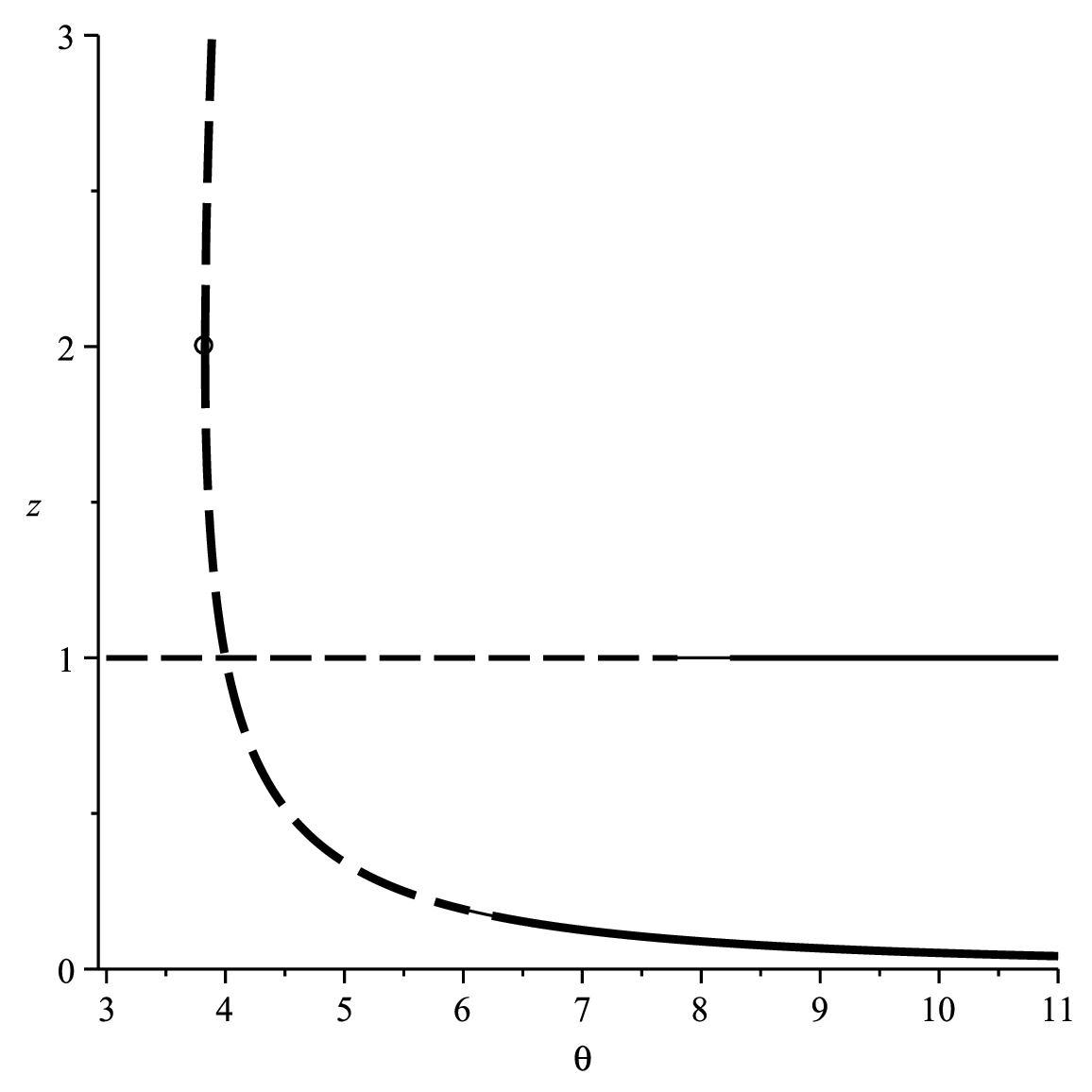}
\caption{The graphs of the functions $z_i=z_i(m,\theta)$, for $q=3$, $m=1$ and the graph of $z(\theta)\equiv 1$. The bold curves correspond to regions of solutions where the corresponding TISGM is non-extreme (corresponds to part (ii), b) of Theorem \ref{tne}). The dashed bold curves correspond to regions of solutions where the corresponding TISGM is extreme (see parts (a) and (b) of Theorem \ref{t1}). The gaps between the two types of curves are given by thin curves.}\label{fn5}
\includegraphics[width=8cm]{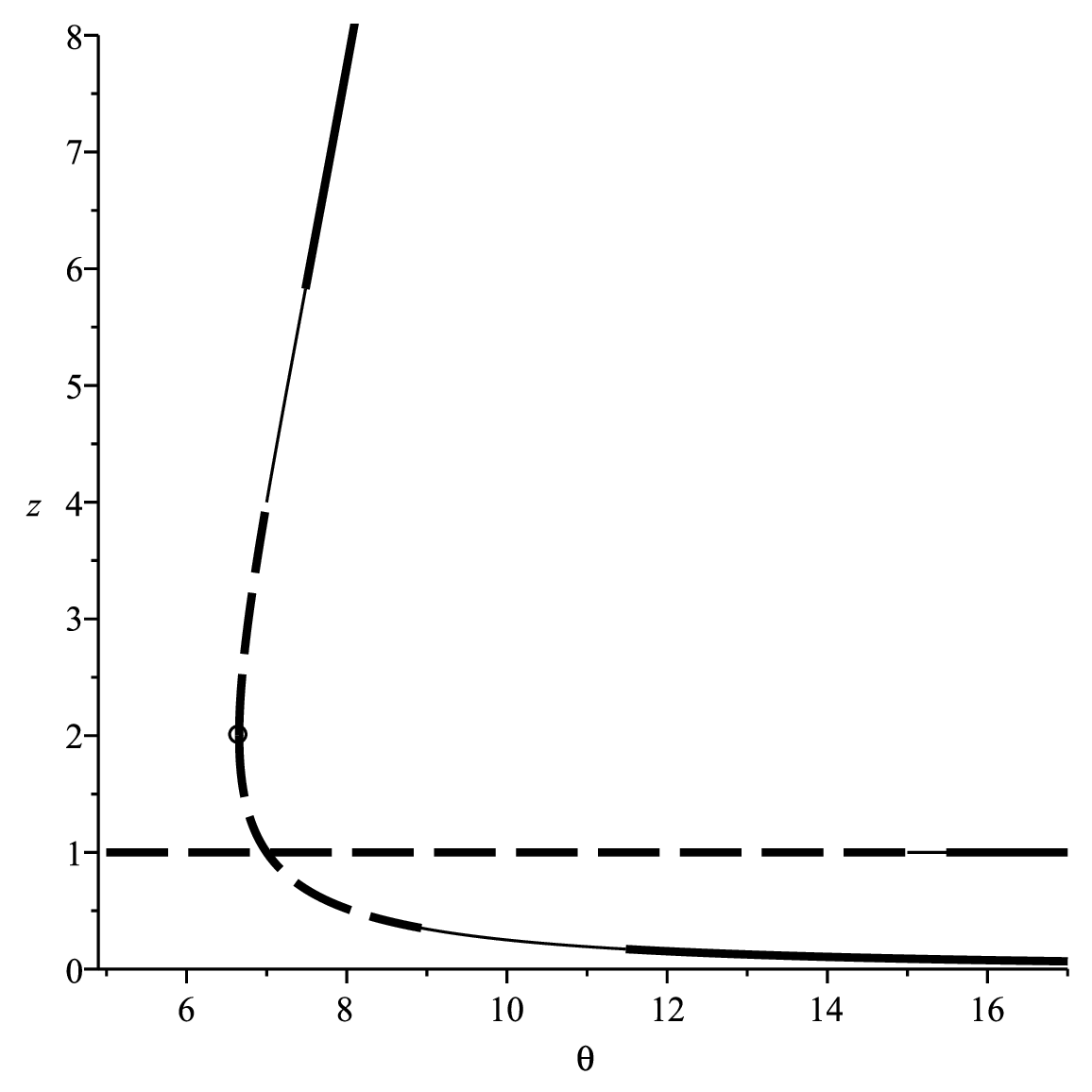}
\caption{The graphs of the functions $z_i=z_i(m,\theta)$, for $q=6$, $m=2$ and the graph of $z(\theta)\equiv 1$. The types of curves corresponding
to certain extremality and certain non-extremality are as in Fig.\ref{fn5} (corresponds to Theorem \ref{t3}).}\label{fn7}
\end{figure}
\begin{figure}
\includegraphics[width=8cm]{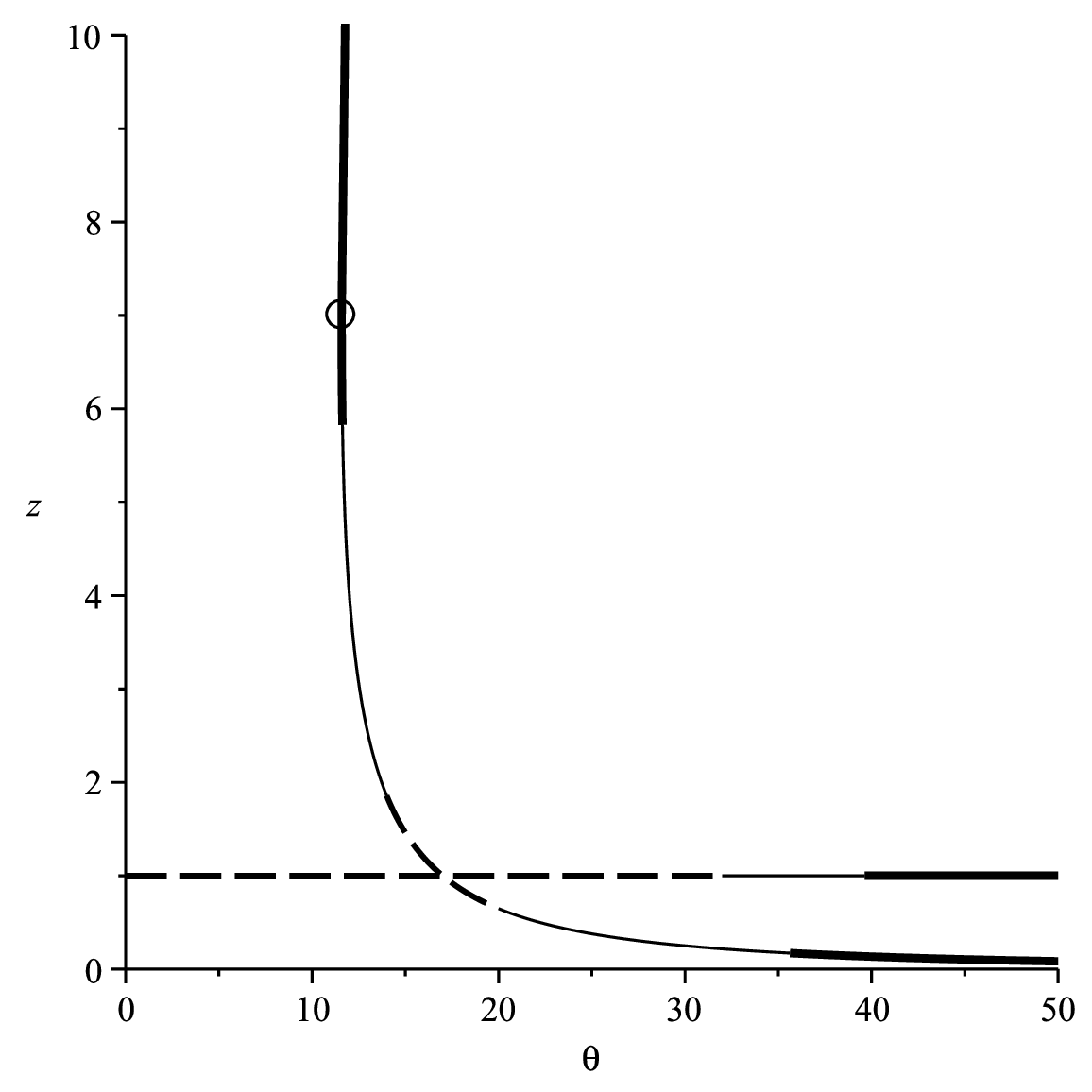}
\caption{The graphs of the functions $z_i=z_i(m,\theta)$, for $q=16$, $m=2$ and the graph of $z(\theta)\equiv 1$. The types of curves are again as in Fig.\ref{fn5} (corresponds to part a) of Theorem \ref{th2}).}\label{fn9}
\end{figure}

{\bf 5.2. Reconstruction insolvability on trees}

We use a result of \cite{MSW} to establish a bound for reconstruction insolvability corresponding to the
matrix (channel) of a solution $z\ne 1$.

Let us first give some necessary definitions from \cite{MSW}. For $k\geq 2$, let $\mathbb T^k$ denote a {\it half tree}, i.e., the infinite
rooted $k$-ary tree (in which every vertex has $k$ children).  Consider an
{\it initial finite complete subtree} $\mathcal T$, that  is a tree of the
following form: in the rooted tree $\mathbb T^k$, take all vertices at distance $\leq d$ from
the root, plus the edges joining them, where $d$ is a fixed constant. We identify subgraphs of $\mathcal T$ with their vertex sets and write $E(A)$ for the edges within a subset $A$ and $\partial A$ for the boundary of $A$, i.e., the neighbors of $A$ in $(\mathcal T\cup \partial\mathcal T)\setminus A$.

In \cite{MSW} the key ingredients are two quantities,
$\kappa$ and $\gamma$, which bound the probabilities of
percolation of disagreement down and up the tree, respectively. Both are properties of
the collection of Gibbs measures $\{\mu^\tau_{{\mathcal T}}\}$, where the boundary condition $\tau$
is fixed and $\mathcal T$ ranges over all initial finite complete subtrees of $\mathbb T^k$.
 For a given subtree $\mathcal T$ of $\mathbb T^k$ and a vertex $x\in\mathcal T$, we write $\mathcal T_x$ for
the (maximal) subtree of $\mathcal T$ rooted at $x$ that is a tree given by
$\mathcal T\cap \mathbb T^k_x$, with $\mathbb T_x^k$ the half tree with root $x$.
We draw the trees with the root at the
top and the leaves at the bottom.
When $x$ is not the root of $\mathcal T$, let $\mu_{\mathcal T_x}^s$
denote the (finite-volume) Gibbs measure in which the
parent of $x$ has its spin fixed to $s$ and the configuration on the bottom boundary  of ${\mathcal T}_x$
(i.e., on $\partial {\mathcal T}_x\setminus \{\mbox{parent\ \ of}\ \ x\}$) is
specified by $\tau$.

 For two measures $\mu_1$ and $\mu_2$ on $\Omega$, $\|\mu_1-\mu_2\|_x$ denotes the variation distance between the projections of $\mu_1$ and $\mu_2$ onto the spin at $x$, i.e.,
$$\|\mu_1-\mu_2\|_x={1\over 2}\sum_{i=1}^q|\mu_1(\sigma(x)=i)-\mu_2(\sigma(x)=i)|.$$

Denote by $\Omega^\tau_\mathcal T$
the set of configurations $\sigma$ given on $\mathcal T\cup \partial\mathcal T$ that agree with $\tau$ on $\partial \mathcal T$,
i.e., $\tau$ specifies
a boundary condition on $\mathcal T$. For any $\eta\in \Omega^\tau_\mathcal T$ and any subset
$A\subseteq \mathcal T$, the Gibbs distribution on $A$ conditional on the configuration outside $A$ being $\eta$ is
denoted by $\mu^\eta_A$.

Let $\eta^{x,s}$ be the
configuration $\eta$ with the spin at $x$ set to $s$.

 Let $\mathbb P=(P_{ij})$ be a channel on $\mathbb T^k$, and $\mu$ its associated Gibbs measure. Following \cite[page 165] {MSW} define
\begin{equation}\label{ka}
\kappa={1\over 2}\max_{i,j}\sum_l|P_{il}-P_{jl}|;
\end{equation}
$$\gamma\equiv\gamma(\mu)=\sup_{A\subset \mathbb T^k}\max\|\mu^{\eta^{y,s}}_A-\mu^{\eta^{y,s'}}_A\|_x,$$
where the supremum is taken over all subsets $A\subset \mathbb T^k$, the maximum is taken over all boundary conditions $\eta$, all sites $y\in \partial A$, all neighbors $x\in A$ of $y$, and all spins $s, s'\in \{1,\dots,q\}$.

As the main ingredient we apply \cite[Theorem 9.3]{MSW}, which is

\begin{thm}\label{tmsw} For an arbitrary (ergodic\footnote{Ergodic means irreducible and aperiodic Markov chain (here irreducible and aperiodic
refers to the probability kernel $\mathbb P$, i.e., the normal Markov chain (where "time" does not have tree structure) with this transition kernel is irreducible and aperiodic). Therefore has a unique stationary distribution $\pi=(\pi_1,\dots,\pi_q)$ with $\pi_i>0$ for all $i$.} and permissive\footnote{Permissive means that for arbitrary finite $A$ and boundary condition outside $A$ being $\eta$ the conditioned Gibbs measure on $A$, corresponding to the channel is positive for at least one configuration on $A$.}) channel ${\mathbb P}=(P_{ij})_{i,j=1}^q$
on a tree, the reconstruction of the corresponding tree-indexed Markov chain
is impossible if $k\kappa\gamma<1$.
\end{thm}

We note that the above mentioned results of \cite{MSW} are given for a rooted Cayley tree (half tree $\mathbb T^k$).
The following lemma says that the results can be extended for unrooted Cayley tree (full tree $\Gamma^k$).

\begin{lemma}\label{lk} Consider a tree-indexed Markov chain with a given fixed transition
matrix $(P_{ij})$ on two different graphs: First, a  Cayley tree where every vertex
has $k+1$ nearest neighbors. Second, a rooted tree where only the root has $k$ nearest neighbors and
all other vertices have $k+1$ neighbors.
Then the chain on the rooted tree is reconstructable if and only if the chain on the
Cayley tree is reconstructable.
\end{lemma}
\begin{proof} Let us denote by $\rho^i_n$ the joint
probability distribution of the spins  $\{\sigma^i(v): v\in V,\, d(x_0,x)=n\}$
which are at the boundary at distance $n$ to the root $x_0$ on the rooted tree,
this is obtained by sending a symbol $i$
from the root $x_0$ using the transition matrix of the Markov chain.
It is well known that the chain is non-reconstructable if and only if,
for any pair of different symbols $i,j\in \{1,\dots, q\}$, the relative entropy
vanishes in the large-$n$ limit, that is we have
$$\lim_{n\rightarrow \infty}S(\rho^{i}_n|\rho^{j}_n)=0.$$
(For more background on this, see Propositions 14 and 15 on page 295 of \cite{EM}).

We will show that non-reconstructability on the rooted tree (the root has $k$ adjacent vertices, the half tree)
implies non-reconstructability on the Cayley tree (the root has $k+1$ adjacent vertices, the full tree).
The other direction is obvious.
To do so, we consider the Cayley tree and single out a vertex $x_0$ which we will call the root.
Let us label its nearest neighbors by $v_1, \dots, v_k, v_{k+1}$. The boundary at distance $n$
for the Cayley tree  decomposes into a first part which is
connected to the root $x_0$ via the vertices
$v_1, \dots,v_k$ and a second part of those vertices
which is connected to $x_0$ via $v_{k+1}$. We denote
the spin distributions corresponding to this decomposition
of boundary vertices by $\rho^{i,<}_n$ and $\rho^{i,>}_n$, and their joint distribution
by  $\rho^{i,\text{full}}_n$. By the Markov property of the chain the spins on these two
parts of the boundary are mutually independent and we can write
$\rho^{i,\text{full}}_n=\rho^{i,<}_n\times \rho^{i,>}_n$. As a consequence we have for
the relative entropy between measures on the full boundary
which are obtained by sending $i\neq j$ that
\begin{equation}\label{juliana1}S(\rho^{i,\text{full}}_n|\rho^{j,\text{full}}_n)
=S(\rho^{i,<}_n|\rho^{j,<}_n)+S(\rho^{i,>}_n|\rho^{j,>}_n).
\end{equation}
We recognize
that $\rho^{i,<}_n$ is the distribution of boundary sites on the rooted tree where the root
has $k$ children.
Noting that $\rho^{i,>}_n$ is the distribution of boundary sites on a tree
where the root has one child, and all other sites have $k$ children,
we use the recursion from root to boundary
to write $\rho^{i,>}_n = \sum_{a=1}^q P_{i,a}\rho^{a,<}_{n-1}$ as a sum over
the spin-value $a$ of the single vertex $v_{k+1}$ at distance $1$ to the root. Here
$\rho^{a,<}_{n-1}$ is the distribution of the boundary sites at distance $n-1$ from
$v_{k+1}$ obtained by sending the letter $a$.
By convexity of the relative entropy in both arguments we get
\begin{equation}\label{juliana2}S(\rho^{i,>}_n|\rho^{j,>}_n)\leq \sum_{a,b}P_{i,a}P_{j,b}
S(\rho^{a,<}_{n-1}|\rho^{b,<}_{n-1}).
\end{equation}
 By \eqref{juliana1} and  \eqref{juliana2} we see now
  that non-reconstruction on the smaller tree, and hence
   $\lim_{n\rightarrow \infty}S(\rho^{i,<}_n|\rho^{j,<}_n)=0$ for all
 $i\neq j$, implies that
$\lim_{n\rightarrow \infty}S(\rho^{i,\text{full}}_n|\rho^{j,\text{full}}_n)=0$ for all $i\neq j$,
and hence non-reconstruction on the larger tree.
 \end{proof}

It is easy to see that the channel ${\mathbb P}$ corresponding to a TISGM of the Potts model is ergodic and permissive. Thus by Theorem \ref{tmsw} and Lemma \ref{lk} the criterion of {\it extremality} of a TISGM is $k\kappa\gamma<1$.

 The constant $\gamma$ does not have a clean general formula, but can be estimated in specific models (as Ising, Hard-Core etc.).
 For example, if ${\mathbb P}$ is the channel of the Potts model then (see Theorem 8.1 in page 160 of \cite{MSW}):
 \begin{equation}\label{bgamma}
\gamma\leq {\theta-1\over\theta+1}.
\end{equation}
Below we shall use this inequality.
Consider the case $z\ne 1$ (where $z=x^2$ and $z$ is a solution to (\ref{rm})), fix a solution of (\ref{pt1}), which has the form $(\underbrace{z,z,\dots,z}_m,1,\dots,1)$ and the corresponding matrix ${\mathbb P}$ given by (\ref{P}).
We shall now compute the constant $\kappa$.
 Recall matrix $\mathbb P$ given in (\ref{P}) and $Z_1$, $Z_2$ given after formula (\ref{P}).
 Since $z$ is a solution to (\ref{rm}) we have $Z_1=\sqrt[k]{z}Z_2$. Using this formula and (\ref{P}) we get
$${1\over 2}\sum_{l=1}^q|P_{il}-P_{jl}|=\left\{\begin{array}{lll}
a, \ \ \ \ \mbox{if} \ \ i,j=1,\dots,m\\[3mm]
b, \ \ \ \ \mbox{if} \ \ i,j=m+1,\dots,q\\[3mm]
c, \ \ \mbox{otherwise},
\end{array}\right.$$
where $a$ and $b$ are defined in (\ref{abl}) and
\begin{equation}\label{c16}
c={1\over 2 Z_1}\Big(z|\theta-\sqrt[k]{z}|+|1-\theta \sqrt[k]{z}|+(z(m-1)+q-m-1)|1-\sqrt[k]{z}|\Big).
\end{equation}

Hence we arrive at
\begin{equation}\label{kab}
\kappa=\left\{\begin{array}{ll}
\max\{b,c\}, \ \ \mbox{if} \ \ m=1\\[3mm]
\max\{a,b,c\} \ \ \mbox{if} \ \ m\geq 2.
\end{array}\right.
\end{equation}

{\bf 5.3. Proof of Theorem \ref{t1}}
\begin{proof} Write the equation (\ref{rm}) for $k=2$, $m=1$
(the conditions of theorem).
First take the square root on both sides,
then obtain a cubic equation for $\sqrt{z}$ which by dividing out the solution
$\sqrt{z} = 1$ can be simplified to $z-(\theta-1)\sqrt{z}+q-1 = 0.$
For the solutions $\sqrt{z_1}$, $\sqrt{z_2}$ of the last equation one has $\sqrt{z_1}+\sqrt{z_2}=\theta-1$.
Using this equality we get
$$\sqrt{z_1}\leq \sqrt{z_1}+\sqrt{z_2}=\theta-1<\theta.$$ Similarly, one gets  $\theta > \sqrt{z_2}$.
From these inequalities it follows that $|\theta-\sqrt{z_i}|=\theta-\sqrt{z_i}$. Since $\theta>1$ for $z_i\geq 1$ we get
$$ |1-\theta \sqrt{z_i}|=\theta\sqrt{z_i}-1, \ \  |1-\sqrt{z_i}|=\sqrt{z_i}-1.$$
We can simplify the expression for $c$ (see (\ref{c16}), for $k=2$, $m=1$) using the above
mentioned equalities for the absolute values present in $c$ and, using that $z_i$ is a solution to
 $z-(\theta-1)\sqrt{z}+q-1 = 0$ (i.e. replacing $z$ by $(\theta-1)\sqrt{z}-q+1$), finally we obtain
\begin{equation}\label{c1}
c={1\over Z_1}\left((\theta+q-2)\sqrt{z_i}-(q-1)\right).
\end{equation}
Recall $b={(\theta-1)\sqrt{z_i}\over Z_1}$ (see (\ref{abl})).
 For $z_i\geq 1$, using $b$ and $c$ given by (\ref{c1}) it is straightforward to see that $b\leq c$ (with equality if and only if $z_i=1$).
Consequently, using (\ref{kab}) under condition $m=1$ we get $\kappa=c$.

(a) {\it Case}: $z_1\geq 1$. For $\theta\in [\theta_1,\theta_c]=[1+2\sqrt{q-1}, q+1]$ we have $z_1\geq 1$.
By (\ref{bgamma}) and $\kappa=c$ we have
$$
2\gamma\kappa\leq 2c{\theta-1\over \theta+1}={2(\theta-1)((\theta+q-2)\sqrt{z_1}-(q-1))\over (\theta z_1+q-1)(\theta+1)}.
$$
Thus, to check $2\gamma\kappa<1$ it is sufficient to check
\begin{equation}\label{ox}
{2(\theta-1)((\theta+q-2)\sqrt{z_1}-(q-1))\over (\theta z_1+q-1)(\theta+1)}<1.
\end{equation}
Since $z_1$ is a solution to $z-(\theta-1)\sqrt{z}+q-1 = 0$ we can
replace in (\ref{ox}) $z_1$ by $(\theta-1)\sqrt{z_1}-q+1$,
simplify and put
 $$\sqrt{z_1}={\theta-1-\sqrt{(\theta-1)^2-4(q-1)}\over 2}.$$
 Then denoting
$$\psi_1(\theta)=(\theta-1)(\theta^2-\theta+6-4q)-(\theta^2-\theta+4-2q)\sqrt{(\theta-1)^2-4(q-1)}$$
we obtain that the inequality (\ref{ox}) is equivalent to $\psi_1(\theta)>0$, $\theta\in [\theta_1,\theta_c)$.

By $(\theta-1)^2>4(q-1)$, i.e. $\theta>\theta_1$ we get
$$\theta^2-\theta+6-4q>0, \ \ \theta^2-\theta+4-2q>0.$$
Using these inequalities the inequality $\psi_1(\theta)>0$ can be simplified to
\begin{equation}\label{u}
u(\theta)=\theta^3-(q-1)\theta^2-(2q-3)\theta+(4q^2-13q+11)>0, \ \ \theta\in [\theta_1,\theta_c].
\end{equation}
The function $u(\theta)$ for $\theta\in (0, \theta_c)$ has a minimum at $\bar\theta=\bar\theta(q)=(q-1+\sqrt{q^2+4q-8})/3$.
Now we calculate the minimal value $u(\bar\theta)$. Since $\bar\theta$ is a solution to
$$u'(\theta)=3\theta^2-2(q-1)\theta-(2q-3)=0$$
in $u(\bar\theta)$ we replace $\bar\theta^2$ by $(1/3)(2(q-1)\bar\theta+(2q-3))$ (in this way we
reduce the powers of $\bar\theta$). We also
replace $\bar\theta^3$ by $(\bar\theta/3)(2(q-1)\bar\theta+(2q-3))$ and simplifying $u(\bar\theta)$
we get
$$u(\bar\theta)=(8-4q-q^2)\bar\theta+(48-56q+17q^2).$$
Now using the formula of $\bar\theta$ from the last equality we get
\begin{equation}\label{ut}
u(\bar\theta)=-q^3+48q^2-156q+136-(q^2+4q-8)^{3/2}.
\end{equation}
From (\ref{ut}) we obtain $u(\bar\theta)>0$ iff $q=3,4,\dots, 16$.
Consequently,
$u(\theta)\geq u(\bar\theta)>0$ (i.e. $\psi_1(\theta)>0$) and $\mu_1(\theta,1)$ is extreme for {\it all} $\theta\in [1+2\sqrt{q-1}, q+1]$, $q=3,4,\dots,16$.

For $q\geq 17$, using the explicit formulas for $\theta_1, \bar\theta, \theta_c$,
one gets $\theta_1<\bar\theta<\theta_c$ and $u(\bar\theta)<0$.
Also we have
$$u(\theta_1)=4(q+2\sqrt{q-1})>0, \ \ u(\theta_c)=4((q-1)^2 +3)>0.$$
Consequently, the equation $u(\theta)=0$ has two solutions $\bar\theta_1$ and $\bar\theta_2$ such that
$$\theta_1<\bar\theta_1<\bar\theta<\bar\theta_2<\theta_c.$$
Thus $u(\theta)>0$ and the measure $\mu_1(\theta,1)$ is extreme for any $\theta\in [\theta_1, \bar\theta_1)\cup (\bar\theta_2,q+1].$\bigskip

{\it Case}: $z_1<1$. In this case, using the explicit formula for $\sqrt{z_1}$
and the condition $\theta>\theta_1$ we get
$\theta \sqrt{z_1}>1$. Indeed, this inequality is equivalent to
$(q-2)\theta^2+\theta+1>0$, which holds for any $\theta>1$.
Consequently, using  $\theta>\sqrt{z_1}$, $\theta \sqrt{z_1}>1$ and $z_1<1$ from formula (\ref{c16}) for $m=1$, for all possible values of $q\geq 2$ we get
$$c={1\over Z_1}(\theta \sqrt{z_1}-1).$$
Recall again $b={(\theta-1)\sqrt{z_1}\over Z_1}$.
Then
$$b\geq c \ \ \Leftrightarrow {(\theta-1)\sqrt{z_1}\over Z_1}\geq {1\over Z_1}(\theta \sqrt{z_1}-1) \ \ \Leftrightarrow
(\theta-1)\sqrt{z_1}\geq \theta \sqrt{z_1}-1 \ \ \Leftrightarrow \sqrt{z_1}\leq 1.$$
Thus if $z_1\leq 1$, i.e., $\theta\in (\theta_c,+\infty)$
then $b\geq c$. Consequently, $\kappa=b$. So we have
$$
2\gamma \kappa\leq 2b\cdot{\theta-1\over \theta+1}.
$$
Hence it suffices to check
\begin{equation}\label{<1}
2b\cdot {\theta-1\over \theta+1}={2(\theta-1)^2\sqrt{z_1}\over (\theta+1)(\theta z_1+q-1)}<1.
\end{equation}

Since $z_1$ is a solution to $z_1-(\theta-1)\sqrt{z_1}+q-1=0$, replacing $z_1$ by  $(\theta-1)\sqrt{z_1}-q+1$ and using the explicit formula for $\sqrt{z_1}$ the inequality
(\ref{<1}) becomes after some simplifications
$$f(\theta)=\theta^3-(q+3)\theta^2-(2q-7)\theta-(q+5)<0.$$
Using the Descartes rule (compare part (iv) of Theorem \ref{tne})
for $f(\theta)$ we see
that the equation $f(\theta)=0$ has at most one positive solution.
Moreover, for any $q\geq 3$ we have $f(q+1)=-4q^2<0$. Now we show that $f(\theta^*)>0$, where $\theta^*=(\sqrt{2}+1)q-1$ (see (\ref{t*}) for $m=1$). Indeed,
we have
$$Q(q)\equiv f(\theta^*)=(4+3\sqrt{2})q^3-6(3+2\sqrt{2})q^2+16(1+\sqrt{2})q-16$$
and
$$Q'(q)=3(4+3\sqrt{2})q^2-12(3+2\sqrt{2})q+16(1+\sqrt{2}).$$
The equation $Q'(q)=0$ has two positive solutions:
$$q_{\pm}={2\over 3}\cdot{9+6\sqrt{2}\pm\sqrt{33+24\sqrt{2}}\over 4+3\sqrt{2}}.$$
We have $q_-<q_+<3$. Therefore $Q(q)$ is an increasing function of $q\geq 3$.
Consequently, for any $q\geq 3$ we have
$$Q(q)\geq Q(3)=21\sqrt{2}-22>0.$$
 So there exists a unique $\theta^{**}$ with $q+1<\theta^{**}<\theta^*$ such that $f(\theta^{**})=0$.
 Thus $f(\theta)<0$ and $\mu_1(\theta,1)$ is extreme for $\theta\in [q+1,\theta^{**})$.
 Summarizing the results for the cases $z_1\geq 1$ and $z_1<1$ we get assertion (a).

 For $q=3$ a numerical analysis shows that $\theta^{**}\approx 6.053$.
 So $\theta^*-\theta^{**}=0.1896$ (see Fig.\ref{fn6}). This is the length of the gap
 between the Kesten-Stigum threshold beyond which non-extremality certainly
 holds and our bound below which extremality certainly holds.

\begin{figure}
\includegraphics[width=8cm]{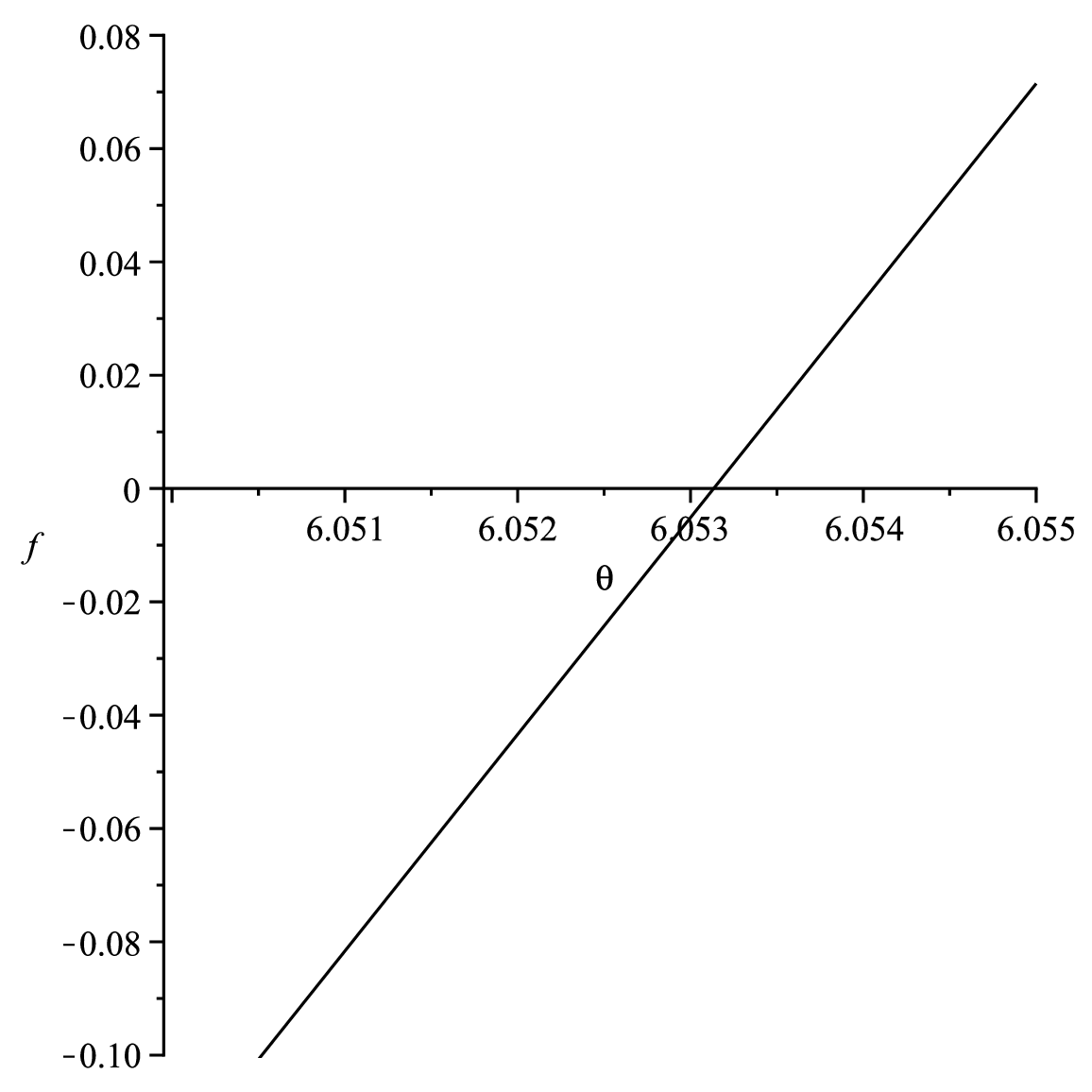}
\caption{The graph of the function $f(\theta)$, for $q=3$, given in the interval
$[6.05,6.055]$ indicates that the solution of $f(\theta^{**})=0$ is $\theta^{**}\approx 6.053$.}\label{fn6}
\end{figure}
(b) For $\theta\geq \theta_1=1+2\sqrt{q-1}$ we have $z_2>1$, i.e $z_2$ is larger than $1$ as soon as it exists.
 Similarly as the case (a) we shall check
$2c{\theta-1\over \theta+1}<1$. Denoting
$$\psi_2(\theta)=(\theta-1)(\theta^2-\theta+6-4q)+(\theta^2-\theta+4-2q)\sqrt{(\theta-1)^2-4(q-1)}$$
and using the fact that $z_2$ is a solution to $z_2-(\theta-1)\sqrt{z_2}+q-1=0$, replacing $z_2$ by  $(\theta-1)\sqrt{z_2}-q+1$ and using the explicit formula  $$\sqrt{z_2}={\theta-1+\sqrt{(\theta-1)^2-4(q-1)}\over 2}$$
after simplifications we get that the last inequality is equivalent to $\psi_2(\theta)>0$, $\theta\geq \theta_1$.
Note that
$$\theta^2-\theta+6-4q\leq \theta^2-\theta+4-2q,$$
so if we prove that
\begin{equation}\label{qq}
\theta^2-\theta+6-4q>0
\end{equation}
then it follows that $\psi_2(\theta)>0$. Since $\theta>1$ the inequality (\ref{qq}) is equivalent to
$\theta>(1+\sqrt{16q-23})/2$. Moreover, one can easily check that
$$(1+\sqrt{16q-23})/2<1+2\sqrt{q-1}=\theta_1$$
Thus we proved that
$\psi_2(\theta)>0$ and $\mu_2(\theta,1)$ is extreme for {\it any} $\theta\geq 1+2\sqrt{q-1}$.
\end{proof}

{\bf 5.4. Proof of Theorem \ref{t3}}
\begin{proof} Write the equation (\ref{rm}) for $k=2$,
first take the square root on both sides,
then obtain a cubic equation for $\sqrt{z}$ which by dividing out the solution
$\sqrt{z} = 1$ can be simplified to $mz-(\theta-1)\sqrt{z}+q-m = 0.$
Using the fact that $z_i$ is a solution to
$mz-(\theta-1)\sqrt{z}+q-m=0$ we simplify the expression of $c$ given in (\ref{c16}) for $m\geq 2$ {\it and} $z=z_i\geq 1$, $i=1,2$.
For the solutions $\sqrt{z_1}$, $\sqrt{z_2}$ of the last equation one has $\sqrt{z_1}+\sqrt{z_2}={\theta-1\over m}$.
Using this equality we get
$$\sqrt{z_1}\leq \sqrt{z_1}+\sqrt{z_2}={\theta-1\over m}<\theta.$$ Similarly, one gets  $\theta > \sqrt{z_2}$.
From these inequalities it follows that $|\theta-\sqrt{z_i}|=\theta-\sqrt{z_i}$. Since $\theta>1$ for $z_i\geq 1$ we get
$$ |1-\theta \sqrt{z_i}|=\theta\sqrt{z_i}-1, \ \  |1-\sqrt{z_i}|=\sqrt{z_i}-1.$$
Consequently, repeatedly  replacing $z_i$ in $c$ by $(1/m)((\theta-1)\sqrt{z_i}-q+m)$ and simplifying we obtain
$$c={1\over m^2Z_1}\left(\{m(\theta^2-\theta+q-m)-(\theta-1)^2\}\sqrt{z_i}-\{(m-1)\theta+1\}(q-m)\right).$$
Recall $a={(\theta-1)z_i\over Z_1}$ (see (\ref{abl})) and by (\ref{kl}) for $z_i>1$ we have $b<a$. Now we shall check that $a>c$:
replacing $z_i$ in $a$ by $(1/m)((\theta-1)\sqrt{z_i}-q+m)$ we get the following sequence of equivalent inequalities
$$a> c \ \  \Leftrightarrow \ \ {(\theta-1)((\theta-1)\sqrt{z_i}-q+m)\over mZ_1} \ \ > $$ $${\{m(\theta^2-\theta+q-m)-(\theta-1)^2\}\sqrt{z_i}-\{(m-1)\theta+1\}(q-m)\over m^2Z_1} \ \ \Leftrightarrow$$
$$[(m+1)(\theta-1)^2-m(\theta^2-\theta+q-m)]\sqrt{z_i}+(m-1)(q-m)(\theta-1)> 0\Leftrightarrow$$
$$[\theta^2-(m+2)\theta-m(q-m-1)+1]\sqrt{z_i}+(m-1)(q-m)(\theta-1)> 0.$$
To show that the last inequality is true for any $\theta>\theta_m$ it suffices\footnote{Note that $(m-1)(q-m)(\theta-1)>0$, because $m\geq 2$, $q>m$ and $\theta>1$} to prove that
\begin{equation}\label{r5}
\theta^2-(m+2)\theta-m(q-m-1)+1>0.
\end{equation}
This quadratic inequality (with respect to $\theta>1$) is true for any
$$\theta>1+{1\over 2}m+{1\over 2}\sqrt{4mq-3m^2}.$$
Now if we prove that
\begin{equation}\label{akc}
 1+{1\over 2}m+{1\over 2}\sqrt{4mq-3m^2}<\theta_m=1+2\sqrt{m(q-m)}
 \end{equation}
then it follows that (\ref{r5}) holds for any $\theta>\theta_m$. So let us check (\ref{akc}):
$$(\ref{akc})\ \ \Leftrightarrow \ \ m+\sqrt{4mq-3m^2}<4\sqrt{m(q-m)}\ \ \Leftrightarrow$$
$$ m^2+2m\sqrt{4mq-3m^2}+4mq-3m^2<16mq-16m\ \ \Leftrightarrow \sqrt{4mq-3m^2}<m+6q-8\ \ \Leftrightarrow$$
$$3q(3q-8)+(m-2)^2+2mq+12>0.$$
The last inequality is true for any $q\geq 2$.
Thus we proved that
\begin{equation}\label{k>}
\kappa=\max\{a,b,c\}=a, \ \ \mbox{if} \ \ z_i>1.
\end{equation}

(i) {\it Case}: $z_1\geq 1$.  If $z=z_1\geq 1$, i.e. $\theta\in [\theta_m,\theta_c]$ then as shown above we
have $\kappa=a$. Hence
$$2\gamma\kappa\leq 2a{\theta-1\over \theta+1}={2(\theta-1)^2 z_1\over ((\theta+m-1)z_1+(q-m))(\theta+1)}.$$
Thus to check  $2\gamma\kappa\leq 1$ it suffices to show
\begin{equation}\label{ox2}
{2(\theta-1)^2 z_1\over ((\theta+m-1)z_1+(q-m))(\theta+1)}<1.
\end{equation}
 The last inequality can be written as
 \begin{equation}\label{os1}
[-\theta^2+(m+4)\theta+m-3]z_1+(\theta+1)(q-m)>0.
\end{equation}
Replacing $z_1$ in (\ref{os1}) by $(1/m)((\theta-1)\sqrt{z_1}-q+m)$ we get
\begin{equation}\label{os}
[-\theta^2+(m+4)\theta+m-3]((\theta-1)\sqrt{z_1}-q+m)+m(\theta+1)(q-m)>0.
\end{equation}
By  formula of $\sqrt{z_1}$ (see (\ref{s})) we have
$$\sqrt{z_1}={\theta-1-\sqrt{(\theta-1)^2-4m(q-m)}\over 2m}$$
$$={\left(\theta-1-\sqrt{(\theta-1)^2-4m(q-m)}\right)\left(\theta-1+\sqrt{(\theta-1)^2-4m(q-m)}\right)\over 2m\left(\theta-1+\sqrt{(\theta-1)^2-4m(q-m)}\right)}$$
\begin{equation}\label{z1c}
={2(q-m)\over \theta-1+\sqrt{(\theta-1)^2-4m(q-m)}}.
\end{equation}
Using formula (\ref{z1c}) of $\sqrt{z_1}$  we obtain from (\ref{os})
\begin{equation}\label{ff}
-\theta^3+(2m+5)\theta^2-7\theta-(2m-3)+(\theta^2-4\theta+3)\sqrt{(\theta-1)^2-4m(q-m)}>0.
\end{equation}
We note that
$$-\theta^3+(2m+5)\theta^2-7\theta-(2m-3)=(\theta-1)[-\theta^2+(2m+4)\theta+2m-3], $$
$$\theta^2-4\theta+3=(\theta-1)(\theta-3).$$
Using these equalities and dividing both side of (\ref{ff}) by $-(\theta-1)$ we get
\begin{equation}\label{tmq}
\theta^2-(2m+4)\theta-2m+3-(\theta-3)\sqrt{(\theta-1)^2-4m(q-m)}<0.
\end{equation}
For $m=2$ from (\ref{tmq}) we obtain
\begin{equation}\label{t2q}
\theta^2-8\theta-1-(\theta-3)\sqrt{(\theta-1)^2-8(q-2)}<0.
\end{equation}
{\it Subcase  $\theta^2-8\theta-1\leq 0$.}  Since $\theta\geq \theta_2=1+2\sqrt{2(q-2)}>3$, the inequality (\ref{t2q}) holds if $\theta^2-8\theta-1\leq 0$.
From the last inequality we get $\theta\leq 4+\sqrt{17}$. Hence we should find $q\geq 4$ such that
$1+2\sqrt{2(q-2)}\leq 4+\sqrt{17}$, i.e. $q\leq 8$. Then for $q<8$ we have $q+1<4+\sqrt{17}$ and the inequality (\ref{t2q}) holds for $\theta$ satisfying
 $$1+2\sqrt{2(q-2)}\leq \theta\leq 4+\sqrt{17}.$$
 For $q=8$ we have
 $$1+2\sqrt{2(q-2)}\leq \theta\leq 4+\sqrt{17}<q+1=9.$$
 Thus we proved that the measure $\mu_1(\theta,2)$ is extreme
 \begin{itemize}
 \item if $q<8$ and $\theta\in [\theta_2,q+1]$
 \item if $q=8$ and $\theta\in [\theta_2,4+\sqrt{17}]$.
  \end{itemize}
 So, for $q=8$ we should check the inequality (\ref{t2q}) for $\theta\in (4+\sqrt{17}, 9]$, since for such $\theta$ we have  $\theta^2-8\theta-1>0$. This
 verification is included in the following subcase.

{\it Subcase  $\theta^2-8\theta-1>0$.} In this case we have $\theta>4+\sqrt{17}$ which,
together with $\theta>\theta_2$, gives that we should consider $q\geq 8$. Rewrite (\ref{t2q}) as
$$\theta^2-8\theta-1<(\theta-3)\sqrt{(\theta-1)^2-8(q-2)}.
$$
Take the square on both sides and simplify up to
$$\varphi(\theta,q)= \theta^3-(q+3)\theta^2+(6q-17)\theta-(9q-19)>0.$$
First let us check this inequality for $q=8$ and $\theta\in [4+\sqrt{17},9]$.
The graph of the function $\varphi(\theta,8)$ is given in Figure \ref{fe10}. Thus we obtain that $\varphi(\theta,8)>0$ and $\mu_1(\theta,2)$ is extreme for $q=8$ on $[4+\sqrt{17},9]$ too.
Summarizing the results of the previous subcase with the present one,
we obtain that this measure is extreme on $[\theta_2,q+1]$ for any $q\leq 8$.
\begin{figure}
\includegraphics[width=8cm]{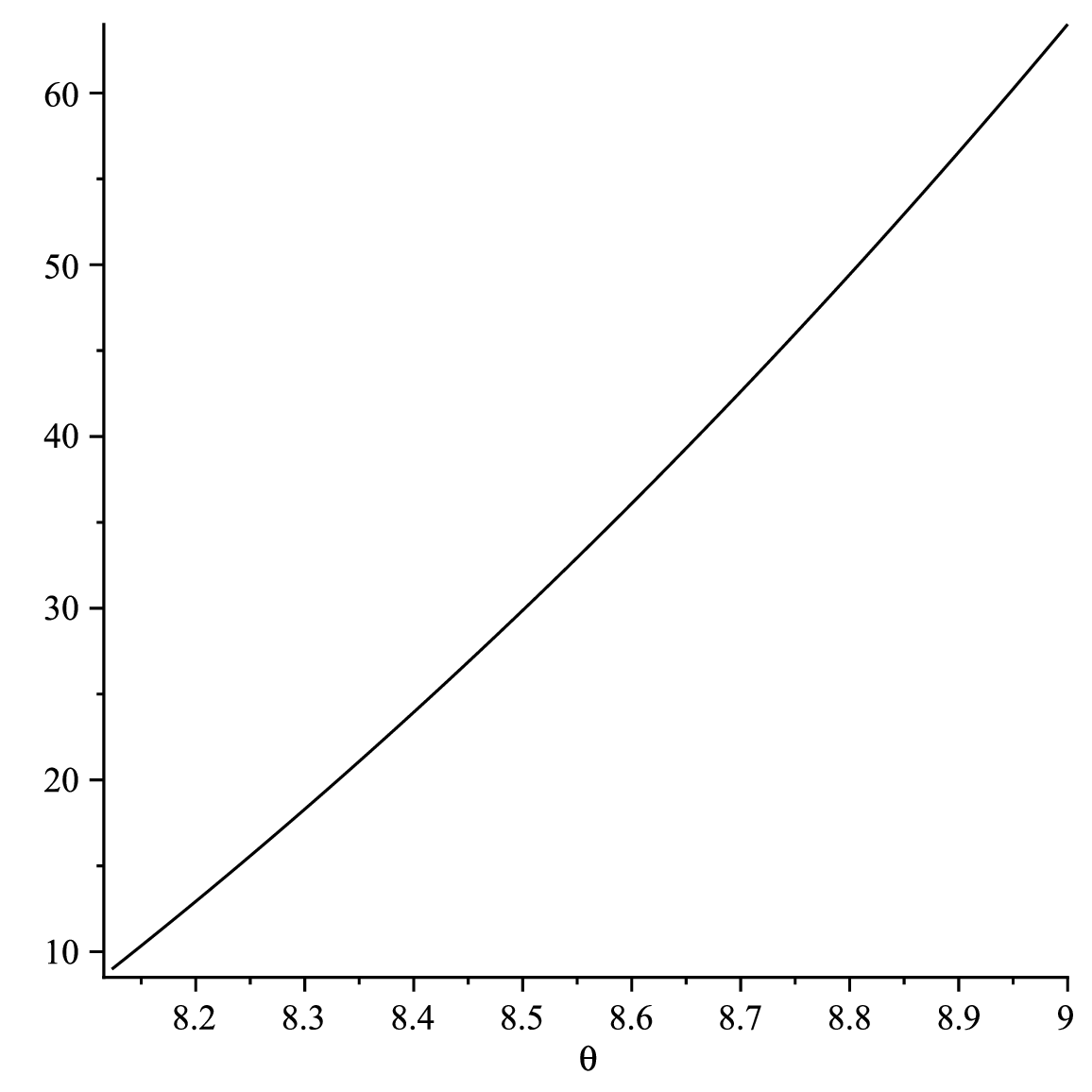}
\caption{The graph of the function $\varphi(\theta,q)$, for $q=8$ and $\theta\in [4+\sqrt{17},9]$.}\label{fe10}
\end{figure}

Now we shall consider the case $q\geq 9$.
The equation $\varphi'(\theta,q)=0$ has two solutions:
$$q_{\pm}={1\over 3}\cdot(q+3\pm\sqrt{q^2-12q+60}).$$
Since the coefficient of $\theta^3$ in $\varphi$ is positive (=1), the function has a maximum at $q_-$ and a minimum at $q_+$. The maximal value is
$$\varphi(q_-,q)=-{2\over 27}[q(q^2-18q+144)-(q^2-12q+60)^{3/2}].$$
Now we prove that $\varphi(q_-,q)<0$ for any $q\geq 9$:
$$\varphi(q_-,q)<0\ \ \Leftrightarrow \ \ q(q^2-18q+144)>(q^2-12q+60)^{3/2} \ \ \Leftrightarrow$$
$$q^2(q^2-18q+144)^2>(q^2-12q+60)^3 \ \ \Leftrightarrow$$
$$\psi(q)=864q^3-15984q^2+129600q-216000>0.$$
The last inequality is true for any $q\geq 9$, indeed $\psi'(q)>0$ and $\psi(9)=285552$.

Consequently, $\varphi(q_+,q)<\varphi(q_-,q)<0$.

For any $q\geq 9$ we have
$$\varphi(\theta_2,q)=-4(2q^2-3q+2(9-3q)\sqrt{2q-4})<0, \ \ \ \ \varphi(q+1,q)=4q(q-6)>0.$$
It follows from these properties that there exists unique $\theta^\dag$ such that $\varphi(\theta^\dag,q)=0$ and the inequality
$\varphi(\theta,q)>0$ holds for any $\theta\in (\theta^\dag,q+1]$, $q\geq 9$ (see Figure \ref{fe11}, for $q=9,10,11$).
\begin{figure}
\includegraphics[width=8cm]{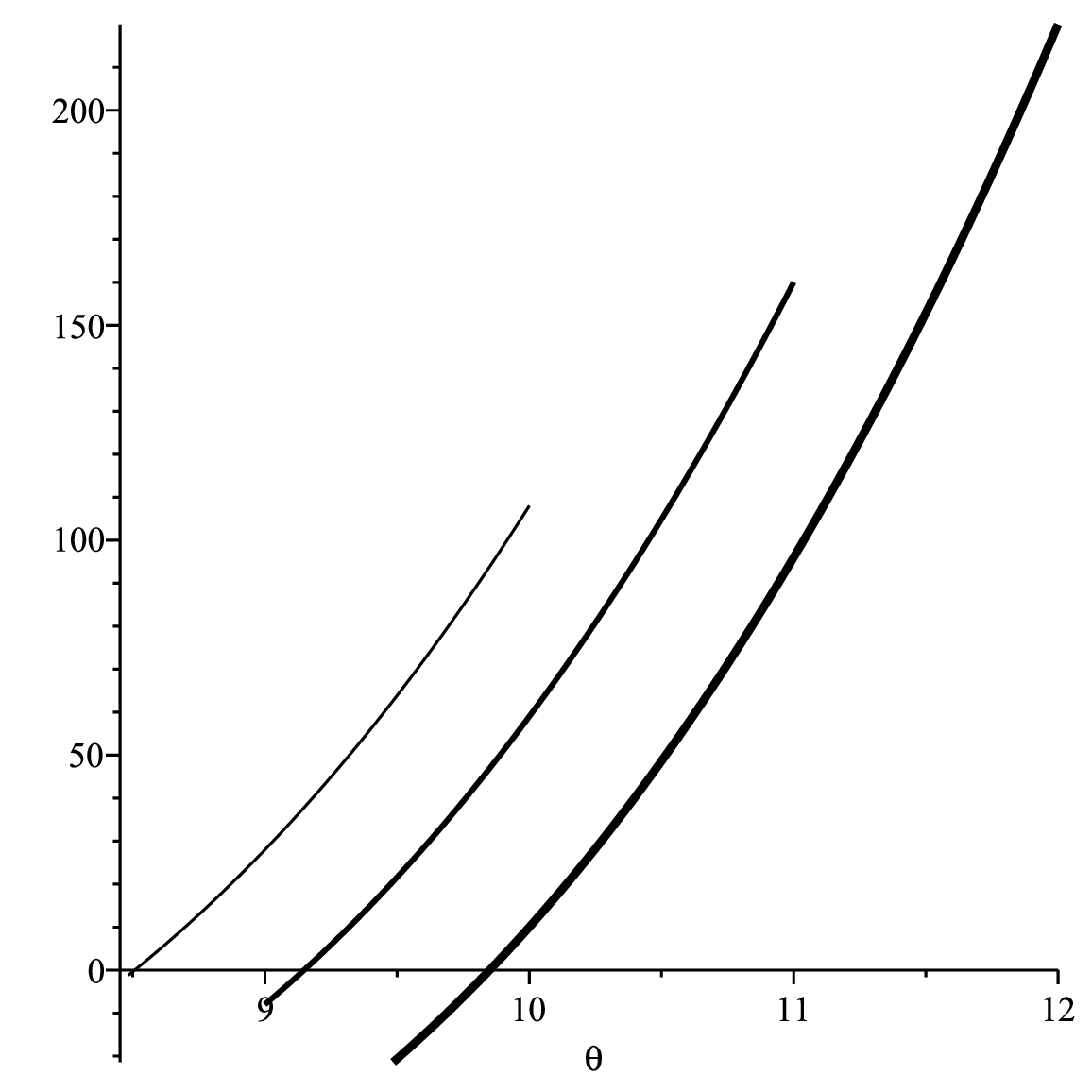}
\caption{The graphs of the functions $\varphi(\theta,q)$, for $q=9$ (left curve), $q=10$ (middle curve) and $q=11$ (right curve), for $\theta\in [\theta_2(q),q+1]$.}\label{fe11}
\end{figure}
Thus we proved the assertions of (i) corresponding
to the case $\theta\leq \theta_c=q+1$.

{\it Case}: $z_1<1$. Using the explicit formula (\ref{z1c}) of $\sqrt{z_1}$
 and the condition $\theta>\theta_m$ we get
$\theta \sqrt{z_1}>1$.  Indeed,
$$\theta \sqrt{z_1}>1\ \ \Leftrightarrow \ \ {2(q-m)\theta\over \theta-1+\sqrt{(\theta-1)^2-4m(q-m)}}>1 \ \ \Leftrightarrow$$
 $$(2(q-m)-1)\theta+1>\sqrt{(\theta-1)^2-4m(q-m)} \ \ \Leftrightarrow$$
$$(q-m-1)\theta^2+\theta+m>0.$$
The last inequality is true for any $\theta>1$.
Consequently, using  $\theta>\sqrt{z_1}$, $\theta \sqrt{z_1}>1$ and $z_1<1$ from
formula (\ref{c16}) replacing
in it $z_1$ by $(1/m)((\theta-1)\sqrt{z_1}-q+m)$,
we get independently on the values of $q$ and $m$ that
$$c={1\over Z_1}\left(\theta\sqrt{z_1}-1\right).$$
For $z_1<1$ we have $a<b$ (see (\ref{kl})). Now, to
prove that $b> c$ for $z_1<1$ (i.e. for $\theta>\theta_c$) we note that
$$ b={(\theta-1)\sqrt{z_1}\over Z_1}> c={1\over Z_1}\left(\theta\sqrt{z_1}-1\right) \ \ \Leftrightarrow $$ $$
(\theta-1)\sqrt{z_1}> \theta\sqrt{z_1}-1\ \ \Leftrightarrow \ \ 1>\sqrt{z_1} \ \ \Leftrightarrow \ \  1>z_1.$$
Consequently, $\kappa=\max\{a,b,c\}=b$ and to check $2\gamma \kappa\leq 1$ it suffices that
\begin{equation}\label{b1}
{2(\theta-1)^2\sqrt{z_1}\over (\theta+1)[(\theta+m-1)z_1+q-m]}<1.
\end{equation}
Again using that $z_1$ is a solution of $mz_1-(\theta-1)\sqrt{z_1}+q-m=0$,  for $m=2$ one can simplify the inequality
(\ref{b1}) to the following
$$(\theta^2-2\theta+5)\sqrt{z_1}-(q-2)(\theta+1)>0.$$
Now, from this inequality, using $\sqrt{z_1}=2(q-2)/[\theta-1+\sqrt{(\theta-1)^2-8(q-2)}]$ (see (\ref{z1c}), for $m=2$) we get
$$g(\theta,q)=\theta^3-(q+3)\theta^2-(2q-15)\theta-(q+13)<0.$$

The equation
$$g'(\theta,q)=3\theta^2-2(q+3)\theta-(2q-15)=0$$
has two solutions
$$\theta_{\pm}=1+{q\over 3}\pm {1\over 3}\sqrt{q^2+12q-36}.$$
Recall $\theta_c=q+1$.  These solutions satisfy
$$\theta_{-}<\theta_+<\theta_c=q+1.$$
Thus $g(\theta,q)$ as a function of $\theta\in [\theta_c,+\infty)$ is an increasing function.
We have $g(\theta_c,q)=-4q(q-2)<0$. Now for $q\geq 4$ we show that $g(\theta^*,q)>0$, where $\theta^*$ is defined in (\ref{t*}):
We compute
$$P(q)\equiv g(\theta^*,q)=(3\sqrt{2}+4)q^3-4(5\sqrt{2}+8)q^2+4(15\sqrt{2}+14)q-112$$
and
$$P'(q)=3(3\sqrt{2}+4)q^2-8(5\sqrt{2}+8)q+4(15\sqrt{2}+14).$$
The equation $P'(q)=0$ has two solutions
$$q_{\pm}={2\over 3}\cdot{10\sqrt{2}+16\pm\sqrt{18+14\sqrt{2}}\over 3\sqrt{2}+4}.$$
These solutions satisfy $q_-<q_+<3$. Thus the function $P(q)$ is increasing for $q\geq 4$. Consequently,
$$P(q)\geq P(4)=16(7\sqrt{2}-9)>0.$$ This completes the proof of $g(\theta^*,q)>0$, for any $q\geq 4$.

 Consequently
$g(\theta,q)=0$  has the unique solution $\breve\theta=\breve\theta(q)$ for
each $q\geq 4$, such that $\theta_c<\breve\theta<\theta^*$. Thus $g(\theta,q)<0$ and $\mu_1(\theta,2)$ is
 extreme for $\theta\in [\theta_c,\breve\theta)$.

 Summarizing the results for the  case $z_1\geq 1$ and the case $z_1<1$
 we get the desired assertions of (i).

 Note that for $q=6$ we have $\theta^*=3(1+2\sqrt{2})\approx 11.485$ and a numerical analysis
 shows that $\breve\theta\approx 8.903$ (see Fig.\ref{fn8}). So $\theta^*-\breve\theta\approx 2.582$.

 \begin{figure}
\includegraphics[width=8cm]{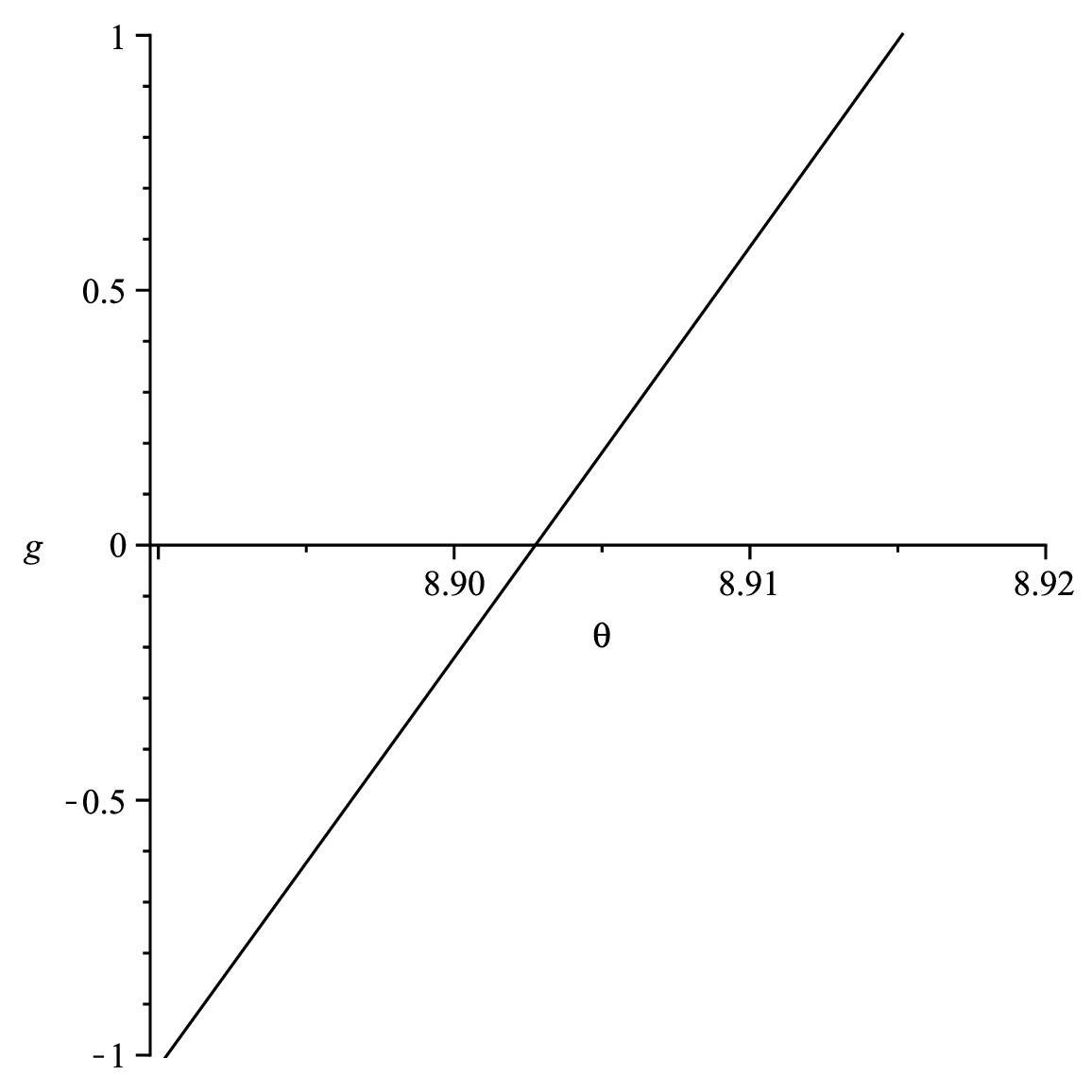}
\caption{The graph of the function $g(\theta,q)$, for $q=6$, given in the small interval $[8.88, 8.92]$ show that the solution of $g(\breve\theta,6)=0$ is $\breve\theta\approx 8.903$.}\label{fn8}
\end{figure}
(ii) Note that $z=z_2>1$, for each $\theta\in [\theta_m, +\infty)$, where $\theta_m=1+2\sqrt{m(q-m)}$. Recall that for $z_2>1$ we have $\kappa=a$ (see (\ref{k>})). Therefore, similarly as in
(\ref{tmq}) we have to check
\begin{equation}\label{tmq2}
\theta^2-(2m+4)\theta-2m+3+(\theta-3)\sqrt{(\theta-1)^2-4m(q-m)}<0.
\end{equation}
For $m=2$ and a given value of $q<9$, from (\ref{tmq2}) we obtain the assertions of (ii).

In the case $q=6$, $m=2$ it is easy to see that the LHS of (\ref{tmq2}) is zero for $\theta=\grave\theta=7$ and the inequality is true for any $\theta\in [\theta_2,7)$. In Theorem \ref{tne} we proved that $\mu_2(\theta,m)$ is non-extreme if $\theta>\widehat\theta$ where
$\widehat\theta$ is defined in (\ref{t*}) which now is $\widehat\theta=6\sqrt{2}-1\approx 7.4852$.  So $\widehat\theta-7\approx 0.4852$. This is the length of the gap between
the Kesten-Stigum bound of non-extremality and our bound of extremality.

(iii) For $\theta=\theta_m$ we have $z_1=z_2>1$ and from (\ref{tmq}) we get
\begin{equation}\label{tmq3}
\theta_m^2-(2m+4)\theta_m-2m+3<0.
\end{equation}
Using the formula $\theta_m=1+2\sqrt{m(q-m)}$ from (\ref{tmq3}) we obtain
\begin{equation}\label{sto}
q<{m+1\over 2m}\left[3m+1+\sqrt{m^2+6m+1}\right], \ \ m\geq 2.
\end{equation}
This completes the proof.
\end{proof}

{\bf 5.5. A continuity argument.}

Let us conclude the paper with softer results which are valid for $\theta$ close to the critical value $\theta_c$
at which the lower branches of the boundary law $z$ degenerate into the free value $z=1$ and
the corresponding Markov chains become close to the free chain.

{\bf Proof of Theorem \ref{th2}}

\begin{proof} a) As it was shown above $\gamma$ and $\kappa$ are bounded by continuous functions of $\theta$. For $\theta=\theta_c=q+1$  we have $z_1=1$, $a=b=c={\theta_c-1\over \theta_c-1+q}={1\over 2}$, $\gamma\leq {\theta_c-1\over \theta_c+1}={q\over q+2}$ and
$\kappa=a=b=c={1\over 2}$. The measure $\mu_0$ (corresponding to the solution $z=(1,1,\dots,1)$) is extreme at $\theta=\theta_c$ if $2\gamma\kappa\leq 1$. We have
$$ 2\gamma\kappa\leq 2a{\theta_c-1\over \theta_c+1}={q\over q+2}< 1, \ \ \mbox{for all} \ \ q\geq 2. $$
Hence $\mu_0$ is extreme at $\theta=\theta_c$. Since $\mu_1(\theta_c,m)=\mu_0$, by the above-mentioned continuity
in a sufficiently small neighborhood of $\theta_c$ the measure $\mu_1(\theta,m)$ is extreme.

b) This follows from Theorem \ref{t3} using the continuity similarly as in case a).
\end{proof}

{\bf Proof of Theorem \ref{th8}}
\begin{proof} By Proposition \ref{pw} we know that for $\theta_{[q/2]}<\theta\ne \theta_c$ there are $2^{q}-1$ TISGMs.
One of them corresponds to $z=1$. Half of the remaining $2^{q}-2$, i.e. $2^{q-1}-1$, TISGMs are generated by
$z_1$ and the other $2^{q-1}-1$ TISGMs are generated by $z_2$. By the above-mentioned results we know that
$\mu_2(\theta,1)$ is extreme. The number of such measures is $q$. By Theorem \ref{th2}, in a sufficiently small punctured neighborhood of $\theta_c$ all measures $\mu_1(\theta,m)$ are extreme. Thus the total number of extremal
TISGM is at least $2^{q-1}+q$.
\end{proof}

\section*{ Acknowledgements}

U.A. Rozikov thanks the  DFG
Sonderforschungsbereich SFB $|$ TR12-Symmetries and Universality in Mesoscopic Systems and the Ruhr-University Bochum (Germany) for financial support and hospitality.  He also thanks the Commission for Developing Countries of the IMU for a travel grant.
We thank all (three) referees for their helpful suggestions. We would like to thank one referee in particular for a very useful remark leading
to an essential simplification in the proofs of Theorem \ref{t1} and Theorem \ref{t3} with improved
bounds.

\end{document}